\documentclass[12pt]{article}
\usepackage{balance}
\usepackage{url}
\usepackage[hidelinks]{hyperref}
\usepackage[utf8]{inputenc}
\usepackage[small]{caption}
\usepackage{graphicx}
\usepackage{amsmath}
\usepackage{amsthm}
\usepackage{booktabs}
\usepackage[noend,ruled,linesnumbered]{algorithm2e}
\usepackage{algorithmic}
\usepackage{xcolor}
\usepackage[shortlabels]{enumitem}
\usepackage[utf8]{inputenc}
\usepackage[english]{babel}
\usepackage{amsmath}
\usepackage{mathtools}
\usepackage{url}
\usepackage{microtype} 
\usepackage[numbers]{natbib}
\usepackage{amsthm}
\usepackage{booktabs}
\usepackage[noend,ruled,linesnumbered]{algorithm2e}
\usepackage{xcolor}
\usepackage[shortlabels]{enumitem}
\usepackage[english]{babel}
\usepackage{tikz}
\usepackage{pgfplots}
\pgfplotsset{compat=1.18}

\newtheorem{theorem}{Theorem}
\newtheorem{lemma}[theorem]{Lemma}

\newcommand{\bk}{\mbox{\bf k}}
\newcommand{\bw}{\mbox{\bf w}}
\newcommand{\bC}{\mbox{\bf C}}
\newcommand{\bT}{\mbox{\bf T}}

\newcommand{\bV}{\mbox{\bf V}}

\newcommand{\bj}{\mbox{\bf j}}

\newcommand{\be}{\begin{equation}}
\newcommand{\ee}{\end{equation}}
\newcommand{\ZZ}{{\bf Z}}

\title{Secure Order Based Voting Using Distributed Tallying}

\author{
Tamir Tassa\thanks{The Open University, Israel. \texttt{tamirta@openu.ac.il}} \and
Lihi Dery\thanks{Ariel University, Israel. \texttt{lihid@ariel.ac.il}} \and
Arthur Zamarin\thanks{The Open University, Israel. \texttt{arthurzam@gmail.com}}
}

\date{}

\begin{document}

\maketitle

\begin{abstract}
\noindent
Electronic voting systems have significant advantages in comparison with physical voting systems. One of the main challenges in e-voting systems is to secure the voting process: namely, to certify that the computed results are consistent with the cast ballots and that the voters' privacy is preserved. We propose herein a secure voting protocol for elections that are governed by order-based voting rules. Our protocol, in which the tallying task is distributed among several independent talliers, offers perfect ballot secrecy in the sense that it issues only the required output while no other information on the cast ballots is revealed. Such perfect secrecy, achieved by employing secure multiparty computation tools, may increase the voters' confidence and, consequently, encourage them to vote according to their true preferences. We implemented a demo of a voting system that is based on our protocol and we describe herein the system’s components and its operation. Our implementation demonstrates that our secure order-based voting protocol can be readily implemented in real-life large-scale electronic elections.
\end{abstract}

\section{Introduction}\label{intro}
Electronic voting has significant advantages in comparison with the more commonplace physical voting:
it is faster, it reduces costs, it is more sustainable,
and in addition, it may increase voters' participation. In recent years, due to the COVID-19 pandemic and the need for social distancing, e-voting platforms have become even more essential.
One of the main challenges in holding e-voting is to secure the voting process:
namely, to make sure that it is secure against attempted manipulations so that the computed results are consistent with the cast ballots and that the privacy of the voters is preserved.

The usual meaning of voter privacy is that the voters remain anonymous. Namely,  
    the linkage between ballots and the voters that cast them remains hidden, thus preserving anonymity. However, the final election results and the full tally (i.e., all ballots) are revealed.  
    While such information exposure may sometimes be considered benign or even desired, in some cases it may be problematic.
    For example, in small-scale elections---such as those for a university senate---publishing the full tally may compromise ballot secrecy. Suppose Bob is a candidate and only a small number of voters participate. If the final tally reveals that Bob received zero votes, and it is known or assumed that Alice supported him, observers (including Bob himself) may infer that Alice did not follow through. Such inference risks may deter voters from expressing their true preferences, either out of fear of social consequences or due to pressure from candidates or peers, particularly in close-knit communities.  
    
    In contrast to weaker privacy notions, such as anonymity or partial tally hiding, perfect ballot secrecy ensures that nothing is disclosed except for the final results. The final results could be just the identity of the winner (say, when selecting an editor-in-chief or a prime minister), $K>1$ winners (e.g., when needing to select $K$ new board members), $K>1$ winners with their ranking (say, when needing to award first, second and third prizes),  or a score for each candidate (as is the case in elections for parliament, where the candidates are parties and the score for each party is the number of seats in the parliament).  
    
    In this paper, we present the design of a tally-hiding voting system that ensures perfect ballot secrecy---a concept sometimes referred to as full privacy (see, e.g.,~\citep{chaum1988}).
    Such secrecy ensures that given any coalition of voters, the protocol does not reveal any information on the ballots beyond what can be inferred from the published results.  
    Such perfect ballot secrecy may reduce the possibility of voters’ coercion and increase their confidence that their vote remains secret. Hence, it may encourage voters to vote according to their true preferences.

There are various families of voting rules that can be used in elections. For example, in {\em score-based} voting rules, a score for each candidate is computed subject to the specifications of the underlying voting rule, and then the winning candidate is the one whose aggregated score is the highest. 
In this study, we focus on {\em order-based} (a.k.a {\em pairwise-comparison}) voting rules, where the relative order of the candidates is considered by the underlying voting rule~\citep{brandt2005decentralized}.

\textbf{Our contributions.} We consider a scenario in which there is a set of voters that hold an election over a given set of candidates, where the election is governed by an order-based voting rule. We devise a fully private protocol for computing the election results. The election output is a ranking of the candidates from which the winning candidate(s) can be determined. Our protocol is lightweight and can be readily implemented in virtual elections.

\medskip
In order to provide privacy for the voters, it is necessary to protect their private data (i.e. their ballots) from the tallier while still allowing the tallier to perform the needed computations on the ballots in order to output the required election results. The cryptographic tool that is commonly used towards this end is {\em homomorphic encryption}, namely, a form of encryption that preserves the algebraic structure and thus enables the performance of meaningful computations on the ciphertexts.
However, homomorphic encryption has a significant computational overhead. The main idea that underlies our suggested system is to use {\em distributed} tallying. Namely, our system involves $D>1$ independent talliers to whom the voters send information relating to their private ballots. With such distributed tallying,
it is possible to replace the costly cryptographic protection shield of homomorphic encryption with the much lighter-weight cryptographic shield of secret sharing (all relevant cryptographic background is provided in Section \ref{background}). The talliers then engage in specially designed protocols of multiparty computation that allow them to validate that each cast ballot 
is a legitimate ballot (in the sense that it complies with the specifications of the underlying voting rule) and then to compute from the cast ballots the required election results, while still remaining totally oblivious to the content of those ballots. A high-level description of our system is presented in Figure \ref{fig: high-level}. The first phase is the voting phase: the voters send to each of the talliers messages (shares) relating to their secret ballots. In the second phase, the talliers communicate among themselves in order to validate each of the incoming ballots (which remain secret to them) and then, at the end of the day, perform the tallying over all legal ballots according to the underlying rule. Finally, the talliers broadcast the results back to the voters. 

\begin{figure*}
    \begin{center}
    \includegraphics[scale=0.4]{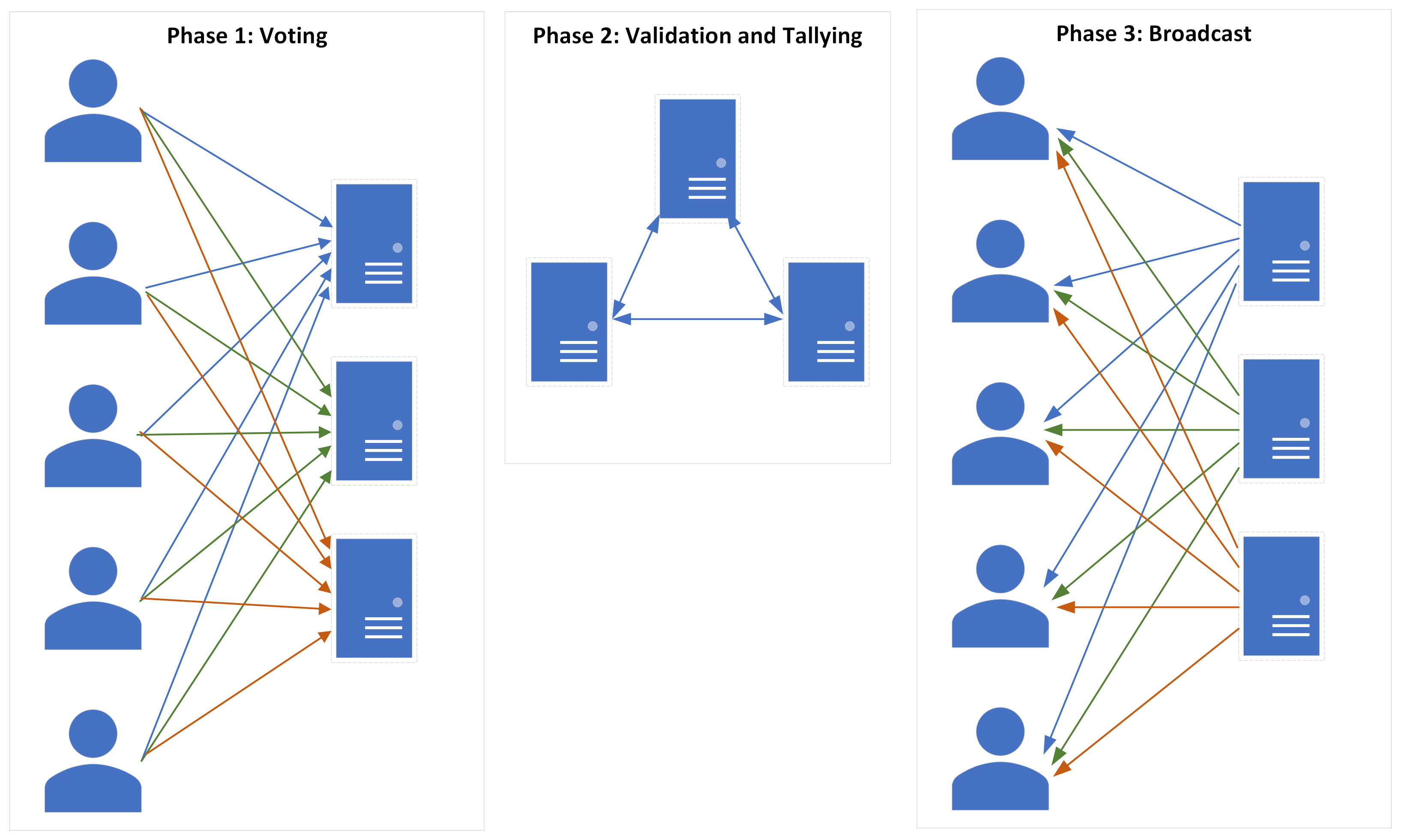}
    \caption{A high-level description of the protocol.}
    \label{fig: high-level}
    \end{center}
\end{figure*}

The bulk of this study is devoted to two order-based voting rules -- {\sc Copeland}~\citep{copeland1951reasonable}, and
{\sc Maximin}~\citep{simpson1969defining} (a.k.a {\sc Kramer-Simpson}), 
which are formally defined in 
Section \ref{formaldef}. In Appendix \ref{other} we describe two other rules in this family, {\sc Kemeny-Young}~\citep{kemeny1959mathematics,young1995optimal} and {\sc Modal Ranking} \citep{caragiannis2014modal}, and describe the extension of our methods for those rules
as well. 

The paper is structured as follows.
In Section \ref{background} we explain the cryptographic principles of distributed tallying.
In Section \ref{semihonest} we present our secure voting protocols
for {\sc Copeland} and {\sc Maximin} rules. We focus there on the case where voters must determine a strict ranking of the candidates. Then, in Section \ref{ties}, we extend our discussion to the more involved case in which the voters' rankings may include ties between some of the candidates.
In Section \ref{demo}
we describe an implementation of a secure voting system that is based on our protocols;  our implementation
is in Python and is open source.
In Section \ref{related} we survey related work. Lastly, Section \ref{conc} provides concluding comments, as well as a description of an implementation of our presented protocol and system in the Gentoo\footnote{https://www.gentoo.org/} council elections. 

The Appendix includes additional discussions. In \ref{scossd} we describe cryptographic sub-protocols that we use in our protocol, while in \ref{other} we discuss secure protocols for two other order-based voting rules $-$ {\sc Kemeny-Young} and {\sc Modal Ranking}.

\section{Distributed tallying}\label{background}
A key principle of our protocol is distributing the tallying task among multiple independent talliers.
This distribution is supported by secret sharing~\citep{Shamir79} and 
secure multiparty computation (MPC)~\citep{yao}.
In Section \ref{ssharing} we give a brief recap of secret sharing.
Then, in Section \ref{principles} we describe the distributed tallying on which our protocol is based. 

Our protocol utilizes MPC sub-protocols that perform secure computations over secret shares. We treat those sub-protocols as black boxes. We describe their inputs and outputs in Section \ref{principles}. Readers seeking a deeper understanding of these MPC components may refer to \ref{scossd}.

\subsection{Secret sharing}\label{ssharing}
Secret sharing schemes~\citep{Shamir79} enable distributing a secret $s$ among a set of parties,
$\bT=\{T_1,\ldots,T_D\}$.
Each party, $T_d$, $d \in [D]:=\{1,\ldots,D\}$,
is given a random value $s_d$, called {\em a share}, that relates to the secret $s$. Those shares satisfy the following two conditions: (a) $s$ can be reconstructed only by combining the shares given to specific {\em authorized} subsets of parties, and (b) shares belonging to unauthorized subsets of parties reveal zero information on $s$.
In {\em threshold} secret sharing there is some threshold $D'\leq D$ and then the authorized subsets
are those of size at least $D'$. Such secret sharing schemes are called $D'$-out-of-$D$. 

Shamir's $D'$-out-of-$D$ secret sharing scheme operates over a finite field $\ZZ_p$, where $p>D$ is a prime sufficiently large so that all possible secrets may be represented in $\ZZ_p$.
It has two procedures: {\sf Share} and {\sf Reconstruct}:

\smallskip
$\bullet$ ${\sf Share}_{D',D}(x)$.
    The procedure samples a uniformly random polynomial $g(\cdot)$ over $\ZZ_p$, of degree at most $D'-1$, where the free coefficient is the secret $s$. That is, $g(x)=s+a_1 x + a_2 x^2 + \ldots + a_{D'-1} x^{D'-1}$, where $a_j$, $1 \leq j \leq D'-1$, are selected independently and uniformly at random from $\ZZ_p$. The procedure outputs $D$ values, $s_d=g(d)$, $d \in [D]$, where $s_d$ is the share given to the $d$th party, $T_d$, $d \in [D]$. 

\smallskip
$\bullet$ ${\sf Reconstruct}_{D'}(s_1,\ldots,s_D)$. The procedure is given any selection of $D'$ shares out of $\{ s_1,\ldots,s_D \}$; it then interpolates a polynomial $g(\cdot)$ of degree at most $D'-1$ using the given point values, and it outputs $s=g(0)$. 
Clearly, any selection of $D'$ shares out of the $D$ shares
will yield the same polynomial $g(\cdot)$ that was used to generate the shares, as $D'$ point values determine a unique polynomial of degree at most $D'-1$. Hence, any selection of $D'$ shares will issue the secret $s$.
On the other hand, any selection of $D'-1$ shares reveals nothing about the secret $s$ (in the sense that every value in the field remains equally probable to be the secret $s$).

\subsection{Distributed tallying and secure computations over secret shares}\label{principles}
Our protocol involves a set of parties, $\{T_1,\ldots,T_D\}$, that are called {\em talliers}.
In the protocol, each voter creates shares of his\footnote{For the sake of simplicity, we keep referring to parties by the pronoun ``he''. In our context, those parties may be voters, who are humans of any gender, or talliers, that are typically (genderless) servers.} private ballot and distributes them to the $D$ talliers.
Since those ballots are matrices over $\ZZ_p$, as we discuss in Section \ref{semihonest}, 
the secret sharing is
carried out for each matrix
entry independently.

We will use Shamir's $D'$-out-of-$D$ secret sharing with $D' := \lfloor (D+1)/2 \rfloor$ (namely, the value $(D+1)/2$ rounded down to the nearest integer). 
With that setting, at least half of the talliers would need to collaborate in order to recover the secret ballots. If the set of talliers is trusted to have an {\em honest majority}, then such a betrayal scenario is impossible. Namely, if more than half of the talliers are honest, in the sense that they would not attempt cheating, then even if all other dishonest talliers collude in attempt to recover the secret ballots from the shares that they hold, they will not be able to extract any information on the ballots. 
Higher values of $D$ (and consequently of
$D' := \lfloor (D+1)/2 \rfloor$)
will imply greater security against coalitions of corrupted talliers, but they will also imply
higher costs.

The first task of the talliers upon receiving the shares of a voter's ballot matrix is to verify that those shares correspond to a legal ballot
(see Section \ref{legal} where we explain what constitutes a legal ballot). 
Then, at the end of the election period, the talliers need to compute the identity of the winning candidates, as dictated by the rule. These tasks would be easy if the talliers could use their shares in order to recover the ballots. However, they must not do so, in order to protect the voters' privacy. Instead, they must perform these computations on the distributed shares, 
without revealing the shares to each other. As we shall see later on, these computations boil down to the four specific tasks that we proceed to describe. 

Let $x_1,\ldots,x_L$ be $L$ secrets in the underlying field $\ZZ_p$ ($L$ is any integer) and assume that the talliers hold $D'$-out-of-$D$ shares of each of those secrets.
Then the talliers can perform the following computational tasks {\em without} recovering the secrets $x_1,\ldots,x_L$ or learning anything about them:
\begin{enumerate}
    \item \textbf{Evaluating arithmetic functions.} If $y=f(x_1,\ldots,x_L)$, where $f$ is some public arithmetic function, compute $D'$-out-of-$D$ shares of $y$.
    \item \textbf{Secure comparisons.} If $y=1_{x_1<x_2}$ is the bit that equals 1 if $x_1<x_2$ and 0 otherwise\footnote{
For any predicate $\Pi$, we denote by $1_{ \Pi }$ the bit that equals 1 if and only if (iff hereinafter) the predicate $\Pi$ holds.}, compute $D'$-out-of-$D$ shares of $y$.
\item \textbf{Testing positivity.} If $p>2N$ (where $N$ is the number of voters) and $x_1 \in [-N,N]$, compute $D'$-out-of-$D$ shares of the bit $1_{x_1>0}$.
\item \textbf{Testing equality to zero.} If $p>2N$ and $x_1 \in [-N,N]$, compute $D'$-out-of-$D$ shares of the bit $1_{x_1=0}$.
\end{enumerate}

All those computational tasks can be carried out by secure multiparty computation (MPC)~\citep{yao} sub-protocols. We treat those computations as black boxes. Namely, it is sufficient to understand what are their inputs and outputs, as described above, but there is no need to understand their internal operation. However, interested readers may refer to \ref{scossd} for more information on those computations.

\section{A secure order-based voting protocol}\label{semihonest}
In this section, we describe our method for securely computing the winners in two order-based voting rules, {\sc Copeland} and {\sc Maximin}.\footnote{Our protocols are designed for scenarios in which it is required to issue just the identity of the winners. In scenarios where more information is needed, such as a ranking of all candidates, or even a score for each candidate, our protocols can easily output also that additional information.} 
We begin with formal definitions in Section \ref{formaldef}.
Then, in Section \ref{legal}, we characterize legal ballot matrices for each of the two rules. Such a characterization is an essential part of our method since the talliers need to verify, in an oblivious manner, that each cast ballot is indeed legal, and does not hide a malicious attempt
to cheat or sabotage the elections. In other words, the talliers need to verify the legality of each cast ballot only through its secret shares, i.e., without getting access to the actual ballot. The characterization that we describe in Section \ref{legal} will be used later on to perform such an oblivious validation.

Then,
in Section \ref{sec11} we introduce our secure voting protocol.
The protocol description includes sub-protocols for
validating the cast ballots and a sub-protocol for computing the final election results from all legal ballots.
The validation sub-protocols are described in Sections \ref{validating} and \ref{validating2}.
The sub-protocols for computing the election results are described in Sections \ref{seccop} and \ref{secmax} for {\sc Copeland} and {\sc Maximin} rules, respectively. 

We conclude this section with a discussion on how to set the size of the underlying field in which all computations take place (Section \ref{ppp}), and a discussion of
the overall security of our protocol (Section \ref{overallsecurity}).

\subsection{Formal definitions}\label{formaldef}
We consider a setting in which there are $N$ voters, $\bV=\{V_1,\ldots,V_N\}$, that wish to hold an election over $M$ candidates, $\bC = \{C_1,\ldots,C_M\}$.
The output of the election is a ranking of $\bC$.
We proceed to define the two order-based rules for which we devise a secure MPC protocol in this section.

$\bullet$ {\sc Copeland.}
Define for each $V_n$ a matrix
$ P_n = (P_n(m, \allowbreak m'): m, m' \in [M])$,
where $P_n(m,m')=1$ if $C_m$ is ranked higher than $C_{m'}$ in $V_n$'s 
ranking,
$P_n(m,m')=-1$ if $C_m$ is ranked lower than $C_{m'}$, and all diagonal entries are 0.
Then the sum matrix, 
\be P = \sum_{n=1}^N P_n \,, \label{pdef}\ee
induces the following score for each candidate:
\be 
\bw(m) := |\{ m' \neq m : P(m,m')> 0\}| + 
\alpha |\{ m' \neq m: P(m,m')= 0 \}|\,.\label{copeland}
\ee
Namely, $\bw(m)$ equals the number of candidates $C_{m'}$ that a majority of the voters ranked lower than $C_m$, plus $\alpha$ times
the number of candidates $C_{m'} $ who broke even with $C_m$.
The parameter $\alpha$ can be set to any rational number between 0 and 1. The most common setting is $\alpha = \frac{1}{2}$; the {\sc Copeland} rule with this setting of $\alpha$ is known as {\sc Copeland}$^\frac{1}{2}$~\citep{faliszewski2009llull}.

\smallskip
$\bullet$ {\sc Maximin.}
Define the matrices $P_n$ so that $P_n(m,m')=1$ if $C_m$ is ranked higher than $C_{m'}$ in $V_n$'s ranking,
while $P_n(m,m')=0$ otherwise.
As in {\sc Copeland} rule, we let $ P$ denote the sum of all ballot matrices, see Eq. (\ref{pdef}). Then  $P(m,m')$ is the number of voters who preferred $C_m$ over $C_{m'}$.
The final score of $C_m$, $m \in [M]$, is then set to
$ \bw(m):=\min_{m' \neq m} P(m,m') $.

In both these order-based rules, we shall refer hereinafter to the matrix $P_n$ as $V_n$'s {\em ballot matrix}, and to $P$ as the aggregated ballot matrix.

\subsection{Characterization of legal ballot matrices}\label{legal}

Here, we characterize the ballot matrices in each of the two order-based rules that we consider. Such a characterization will be used later on in the secure voting protocol.

\begin{theorem}\label{thm1}
An $M \times M$ matrix $Q$ is a valid ballot under the {\sc Copeland} rule iff it satisfies the following conditions:
\begin{enumerate}
    \item $Q(m,m') \in \{-1,1\}$ for all $1 \leq m < m' \leq M$;
    \item $Q(m,m)=0$ for all $m \in [M]$;
    \item $Q(m',m)+Q(m,m')=0$ for all $1 \leq m < m' \leq M$;
    \item If $Q_m:=\sum_{m' \in [M]} Q(m,m')$ then $Q_m \neq Q_{m'}$ for all $1 \leq m<m' \leq M$.
\end{enumerate}
\end{theorem}

\begin{theorem}\label{thm2}
An $M \times M$ matrix $Q$ is a valid ballot under the {\sc Maximin} rule iff it satisfies the following conditions:
\begin{enumerate}
    \item $Q(m,m') \in \{0,1\}$ for all $1 \leq m < m' \leq M$;
    \item $Q(m,m)=0$ for all $m \in [M]$;
    \item $Q(m',m)+Q(m,m')=1$ for all $1 \leq m < m' \leq M$;
    \item If $Q_m:=\sum_{m' \in [M]} Q(m,m')$ then $Q_m \neq Q_{m'}$ for all $1 \leq m<m' \leq M$.
\end{enumerate}
\end{theorem}

Conditions 1-3 in both theorems are clear. As for condition 4, it states that the pairwise relations that $Q$ defines are consistent with a linear ordering of the $M$ candidates. As a counter-example, consider the following {\sc Copeland} ballot matrix over $M=3$ candidates:
$$ Q= \left( \begin{array}{rrr} 0 & 1 & -1  \\ -1 & 0 & 1 \\
 1 & -1 &  0
\end{array} \right) \,. $$
It complies with conditions 1-3 in Theorem \ref{thm1}, but violates condition 4 since $Q_1=Q_2=Q_3=0$. Hence, it is an illegal matrix.
Indeed, the matrix states that $C_1>C_2$ (since $Q(1,2)=1$), $C_2>C_3$, and $C_3>C_1$. Such cyclicity is impossible in a linear ordering.  

\medskip
{\bf Proof of Theorem \ref{thm1}.}
Assume that the ordering of a voter over the set of candidates $\bC$ is $(C_{j_1},\ldots,C_{j_M})$ where $\bj:=(j_1,\ldots,j_M)$ is some permutation of $[M]$. Then the ballot matrix that such an ordering induces is 
\[
     Q=(Q(m,m'))_{1 \leq m,m' \leq M}    
\]
where $Q(m,m')=1$ if $m$ appears before $m'$ in the sequence $\bj$, 
$Q(m,m')=-1$ if $m$ appears after $m'$ in $\bj$, and 
$Q(m,m)=0$ for all $m \in [M]$. Such a matrix clearly satisfies conditions 1, 2 and 3 in the theorem. It also satisfies condition 4, as we proceed to show. Fix $m \in [M]$ and let $k \in [M]$ be the unique index for which $m=j_k$. Then
the $m$th row in $Q$ consists of exactly $k-1$ entries that equal $-1$, $M-k$ entries that equal $1$, and a single entry on the diagonal that equals 0. Hence, $Q_m$, which is the sum of entries in that row, equals $M+1-2k$. Clearly, all those values are distinct, since the mapping $m \mapsto k$ is a bijection. That completes the first part of the proof: every legal {\sc Copeland} ballot matrix satisfies conditions 1-4.

Assume next that $Q$ is an $M \times M$ matrix that satisfies conditions 1-4. Then for each $m \in [M]$, the $m$th
row in the matrix consists of a single 0 entry on the diagonal where all other entries are either $1$ or $-1$. Assume that the number of $-1$ entries in the row equals $k(m)-1$, for some $k(m) \in [M]$, while the number of $1$ entries equals $M-k(m)$. Then the sum $Q_m$
of entries in that row equals $M+1-2k(m)$. As, by condition 4,
all $Q_m$ values are distinct, then $k(m)\neq k(m')$ when $m \neq m'$. Stated otherwise, the sequence $\bk:=(k(1),\ldots,k(M))$ is a permutation of $[M]$. Let $\bj:=(j_1,\ldots,j_M)$ be the inverse permutation of $\bk$; i.e., for each $m \in [M]$, $j_{k(m)}=m$. Then it is easy to see that the matrix $Q$ is the
{\sc Copeland} ballot matrix that corresponds to the ordering $\bj$. That completes the second part of the proof: every matrix that satisfies conditions 1-4 is a legal {\sc Copeland} ballot matrix. 
$\qed$

\smallskip
The proof of Theorem \ref{thm2} is similar to that of Theorem \ref{thm1} and thus omitted.

\medskip
We conclude this section with the following observation. Let us define a projection mapping $\Gamma:\ZZ_p^{M \times M} \mapsto \ZZ_p^{M(M-1)/2}$, which takes an $M \times M$ matrix $Q \in \ZZ_p^{M \times M}$ and outputs its upper triangle,
\be \Gamma(Q):=(Q(m,m'): 1 \leq m < m' \leq M) \,.\label{GammaDef}\ee
Conditions 2 and 3 in Theorems \ref{thm1} and \ref{thm2} imply that every ballot matrix, $P_n$, is fully determined by its upper triangle, $\Gamma(P_n)$, in either of the two voting rules that we consider.

 \subsection{The protocol}\label{sec11}
Here we present Protocol \ref{algvote}, a privacy-preserving implementation of
the {\sc Copeland} and {\sc Maximin} order-based rules.
The protocol computes, in a privacy-preserving manner, the winners in elections that are governed by those rules.
It has two phases. We begin with a bird's-eye view of those two phases, and afterward, we provide a more detailed explanation.

Phase 1 (Lines 1-8) is the voting phase. In that phase, each voter $V_n$, $n \in [N]:=\{1,\ldots,N\}$, constructs his ballot matrix, $P_n$ (Line 2), and then creates and distributes to all talliers corresponding $D'$-out-of-$D$ shares, with $D'=\lfloor (D+1)/2 \rfloor$,
as described in Section \ref{ssharing} (Lines 3-7).
Following that, the talliers jointly verify the legality of the shared ballot (Line 8). 
Phase 2 of the protocol (Lines 9-11) is carried out after the voting phase had ended. In that second phase, the talliers perform an MPC sub-protocol on the distributed shares in order to find 
the winning candidates, while remaining oblivious to the actual ballots.

After constructing his ballot matrix in Line 2, voter $V_n$, $n \in [N]$, samples a random share-generating polynomial of degree $D'-1$ for each of the $M(M-1)/2$ entries in $\Gamma(P_n)$, where $\Gamma$ is the projection mapping defined in Section \ref{legal} (Lines 3-5). Then, $V_n$ sends each tallier $T_d$ his relevant share of each of those entries, namely, the value of the corresponding share-generating polynomial at $x=d$, 
$d \in [D]$ (Lines 6-7).

In Line 8, the talliers engage in an MPC sub-protocol to verify the legality of $V_n$'s ballot, without actually recovering that ballot.
There are two validation sub-protocols that need to be executed and we describe them in Sections \ref{validating} and \ref{validating2}. Ballots that are found to be illegal are discarded. (In such cases it is possible to notify the voter that his ballot was rejected and allow him to resubmit it.)

Phase 2 (Lines 9-11) takes place at the end of the voting period, 
after the voters had cast their
ballots.
First (Lines 9-10), each of the talliers, $T_d$, $d \in [D]$, computes his $D'$-out-of-$D$ share, denoted $G_d(m,m')$, in the $(m,m')$th entry of the aggregated ballot matrix $P$, see Eq. (\ref{pdef}), for all $1 \leq m<m' \leq M$. This computation follows from the linearity of the secret-sharing scheme. Indeed, as 
$g_{n,m,m'}(d)$ is $T_d$'s share of $P_n(m,m')$ in a $D'$-out-of-$D$ Shamir's secret sharing, then 
$G_d(m,m') = \sum_{n \in [N]} g_{n,m,m'}(d)$ is a 
$D'$-out-of-$D$ Shamir's secret share of $P(m,m') = \sum_{n \in [N]} P_n(m,m')$.

Then, in Line 11, the talliers engage in an MPC
sub-protocol in order to find the indices of the $K$ winning candidates.
How can the talliers do that when none of them actually holds
the matrix $P$? We devote our discussion in Sections \ref{seccop} and \ref{secmax} to that interesting computational challenge.

\SetAlgorithmName{Protocol}{}{}
\begin{algorithm}[htb]
\SetAlgoLined
\DontPrintSemicolon
    \SetKwInput{Input}{Input}
    \Input{A set of $M$ candidates $\bC$; $K \in [M]$; a set of voters $\bV$.}
\ForAll{$V_n$, $n \in [N],$}{
Construct the ballot matrix, $P_n$, according to the selected indexing of candidates and the voting rule

\ForAll{$1 \leq m < m' \leq M$}{
Select uniformly at random $a_{n,m,m',j} \in \ZZ_p$, $1 \leq j \leq D'-1$

Set $g_{n,m,m'}(x) = P_n(m,m') + \sum_{j=1}^{D'-1}
a_{n,m,m',j}x^j$
}
\ForAll{$d\in [D]$}{
Send to $T_d$ the set $\{ (n,m,m',g_{n,m,m'}(d)): 1 \leq m<m' \leq M\}$
}
After all talliers receive their shares of $V_n$'s ballot, they engage in an MPC sub-protocol to check its legality
}
\ForAll{$T_d$, $d \in [D]$}{
Set $G_d(m,m') = \sum_{n \in [N]} g_{n,m,m'}(d)$, for all $1 \leq m<m' \leq M$
}
$T_1,\ldots,T_D$ find the indices of the $K$ winners and publish them

\SetKwInOut{Output}{Output}
    \Output{The $K$ winning candidates from $\bC$.}
\caption{A basic protocol for secure order-based voting}
\label{algvote}
\end{algorithm}

\subsection{Verifying the legality of the cast ballots}\label{validating}
Voters may attempt to cheat by submitting illegal ballots in order to help their preferred candidate.
For example, a dishonest voter may send the talliers shares of the matrix $cQ$,
where $Q$ is his true ballot matrix and $c$ is an ``inflating factor" greater than 1. 
If such a cheating attempt remains undetected, then that voter would manage to multiply his vote by a factor of $c$. The shares of the illegal ``inflated" matrix are indistinguishable from the shares of the legal
ballot matrix (unless, of course, a majority of $D'$ talliers collude and recover the ballot $-$ a scenario that is impossible under our assumption that the talliers have an honest majority).
Hence, it is necessary to devise a mechanism that would enable the talliers to check that the shares they received from each of the voters correspond to a legal ballot matrix, without actually recovering that matrix.

Malicious voters can sabotage the elections also in other manners. For example, a voter may create a legal ballot matrix $Q$ and then alter the sign of some of the $\pm 1$ entries in the shared upper triangle $\Gamma(Q)$, so that the resulting matrix no longer corresponds to an ordering of the candidates. 
Even though such a manipulated
matrix cannot serve a dishonest voter in an attempt to help a specific candidate, it still must be detected and discarded. Failing to detect the illegality of such a ballot may result in an aggregated matrix $P$ that
differs from the aggregated matrix $P'$ that corresponds to the case in which that ballot is rejected. In such a case, the set of winners that $P$ would determine may differ from that which $P'$ determines. In summary, it is essential to validate each of the cast ballots in order to prevent any
undesirable sabotaging of the elections.

In real-life
electronic elections where voters typically cast their ballots on certified computers in voting centers, the chances of hacking such computers and tampering with the software that they run are small. However, for full-proof security and as a countermeasure against dishonest voters that might manage to hack the voting system, we proceed to describe an MPC solution
that enables the talliers to verify the legality of each ballot, even though those ballots remain hidden from them.
We note that in case a ballot is found to be illegal, the talliers may reconstruct it (by means of interpolation from the shares of $D'$ talliers) and use the recovered ballot as proof of the voter's dishonesty.
The ability to construct such proofs, which could be used in legal proceedings against dishonest voters, might deter voters from attempting cheating
in the first place.

We proceed to explain how the talliers can verify the legality of the cast ballots in each of the two order-based rules. That validation is based on the characterizations of legal ballots as provided in Theorems \ref{thm1} and \ref{thm2} for {\sc Copeland} and {\sc Maximin}, respectively. Note that the talliers need only to verify conditions 1 and 4 (in either Theorem \ref{thm1} or \ref{thm2});
condition 2 needs no verification since the voters do not distribute shares of the diagonal entries, as those entries are known to be zero; and condition 3 needs no verification since the voters distribute shares only in the upper triangle and then the talliers use condition 3 in order to infer the lower triangle from the shared upper triangle.

\subsubsection{Verifying condition 1}\label{vercond1}
Consider the shares that a voter distributed in $\Gamma(Q)$, where $Q$ is his ballot matrix.
The talliers need to verify that each entry in the shared $\Gamma(Q)$
is either $1$ or $-1$ in {\sc Copeland}, or either $1$ or $0$ in {\sc Maximin}.
The verification is performed independently on each of the $M(M-1)/2$ entries of the shared upper triangle. A shared scalar $x$ is in $\{-1,1\}$ (resp. in $\{0,1\}$) iff $(x+1)\cdot (x-1)=0$ (resp. $x\cdot (x-1)=0$).
Hence, the talliers use their shares of $x$ in order to compute the product $(x+1)\cdot (x-1)$ for {\sc Copeland} or
$x \cdot (x-1)$ for {\sc Maximin}, using the procedure for evaluating arithmetic functions (as described in Section \ref{principles} and further discussed in Section \ref{mpc} in the Appendix).
If the computed results are zero for all $M(M-1)/2$ entries of 
$\Gamma(Q)$,
then $\Gamma(Q)$ satisfies condition 1 in Theorem \ref{thm1}
({\sc Copeland}) or in Theorem \ref{thm2} ({\sc Maximin}).
If, on the other hand, at least one of the results is nonzero,
then the ballot will be rejected.

\subsubsection{Verifying condition 4}

As the talliers received $D'$-out-of-$D$ shares of each entry in $\Gamma(Q)$, they can compute 
$D'$-out-of-$D$ shares of the corresponding 
row sums, 
$Q_m$, $m \in [M]$, as we proceed to show. 
In {\sc Copeland},
conditions 2 and 3 in Theorem \ref{thm1} imply that
\be Q_m = \sum_{m' \in [M]} Q(m,m') = \sum_{m' >m} Q(m,m') -
\sum_{m'<m} Q(m',m) \,. \label{qmcop}\ee
Since the talliers hold shares of $Q(m',m)$ for all $1 \leq m'<m \leq M$, they can use Eq. (\ref{qmcop}) and the linearity of secret sharing to compute shares of $Q_m$, $m \in [M]$. 
In {\sc Maximin}, on the other hand, conditions 2 and 3 in Theorem \ref{thm2}
imply that
\be Q_m = \sum_{m' >m} Q(m,m') -
\sum_{m'<m} Q(m',m) + (m-1)\,. \label{qmmax}\ee
Here too, the linearity of secret sharing
and the relation in Eq. (\ref{qmmax}) enable the talliers to compute shares of $Q_m$, $m \in [M]$, also in the case of {\sc Maximin}.

Now, it is necessary to verify that all $M$ values $Q_m$, $m \in [M]$,
are distinct. That condition can be verified by computing the product
\be F(Q):= \prod_{1 \leq m'<m \leq M} (Q_m - Q_{m'})  \,. \label{prodq}\ee 
Condition 4, in both rules, holds iff $F(Q) \neq 0$.
Hence, the talliers, who hold shares of $Q_m$, $m \in [M]$, may compute $F(Q)$ and then accept the ballot iff $F(Q) \neq 0$. 

\subsubsection{Privacy}\label{privacy}
A natural question that arises is whether the above described validation process poses a risk to the privacy of the voters.
In other words, a voter that casts a legal ballot wants to be ascertained that the validation process only reveals that the ballot is legal, while all other information is kept hidden from the talliers. We proceed to examine that question.

The procedure for verifying condition 1 in Theorems \ref{thm1} and \ref{thm2} offers perfect privacy for honest voters.
If the ballot $Q$ is legal then all computed values will be zero.
Hence, apart from the legality of the ballot, the talliers will not learn anything about the content of the ballot. 

The procedure for verifying condition 4
in Theorems \ref{thm1} and \ref{thm2}
offers {\em almost} perfect privacy, in the following sense.
If $Q$ is a valid ballot in {\sc Copeland}, then the ordered tuple
$(Q_1,\ldots,Q_M)$ is a permutation of
the ordered tuple
$(-M+1,-M+3,\ldots, M-3,M-1)$. This statement follows from the proof of condition 4 in Theorem \ref{thm1}.
Hence, as can be readily verified,
if $Q$ is a legal ballot then
the value of $F(Q)$, Eq. (\ref{prodq}), which the talliers compute in the validation procedure, equals
\be F(Q)= \pm 2^{\binom{M}{2}} \cdot \prod_{1 \leq m'<m \leq M} (m - m')\,, \label{prodq1}\ee
where the sign of the product in Eq. (\ref{prodq1}) 
is determined by the signature of $(Q_1,\ldots,Q_M)$ when viewed as a permutation of
the ordered tuple
$(-M+1,-M+3,\ldots, M-3,M-1)$.
Hence, since a ballot in {\sc Copeland} describes some ordering (permutation) of the candidates, the talliers will be able to infer the signature of that permutation, but nothing beyond that.
To eliminate even that negligible leakage of information, the talliers will simply compute $F(Q)^2$ instead of $F(Q)$. 

Similarly, the procedure for verifying condition 4 in Theorem \ref{thm2}, for {\sc Maximin}, is also privacy-preserving in the same manner. Indeed, in the case of {\sc Maximin}, 
$(Q_1,\ldots,Q_M)$ is a permutation of
the ordered tuple
$(0,1,\ldots,M-2,M-1)$,
and, therefore,
$$ F(Q)= \pm  \prod_{1 \leq m'<m \leq M} (m - m')\,, $$
where the sign of the above product 
is determined by the signature of $(Q_1,\ldots,Q_M)$ 
as a permutation of
$(0,1,\ldots,M-2,M-1)$. Here too, the talliers will recover $F(Q)^2$ instead of $F(Q)$ to ensure perfect privacy.

\subsection{Verifying the legality of the secret sharing}\label{validating2}
A malicious voter $V$ may attempt to sabotage the election by distributing to the $D$ talliers shares that do not correspond to a polynomial of degree $D'-1$.
Namely, if $\{s_1,\ldots,s_D\}$ are the shares that $V$ distributed to $T_1,\ldots,T_D$ in one of the entries in his ballot matrix, it is possible that there is no polynomial $g$ of degree (up to) $D'-1$ such that $g(d)=s_d$ for all $d \in [D]$.
By carefully selecting those shares, they may still pass the verification tests as described in Section \ref{validating} with some non-negligible probability, and then they would be integrated in the final computation of the winners. Since such shares do not correspond to a legal vote, they may sabotage the final computation of the winners (as we describe later on in Sections \ref{seccop} and \ref{secmax}).
To prevent such an attack, we explain herein how the talliers may detect it without learning anything on the submitted ballot
(beyond the mere legality of the secret sharing that was applied to it). 

Sub-protocol \ref{verify_ss} performs the desired testing. The talliers apply it on the secret shares $\{s_1,\ldots,s_D\}$ that they had received in any of the entries $s=Q(m',m)$ in an incoming ballot matrix $Q$.

First, each tallier $T_d$, $d \in [D]$, produces a random number $r_d$ and distributes to all talliers $D'$-out-of-$D$ shares of it (Lines 1-3). Then, each tallier $T_d$ adds to his share, $s_d$, the shares received from all talliers in the random numbers that they had produced, and broadcasts the result, denoted $\hat{s}_d$ (Lines 4-6).
The talliers proceed to find, by means of interpolation, a polynomial $f$ that satisfies $f(d)=\hat{s}_d$ for all $d \in [D]$ (Line 7). If the degree of the polynomial is $D'-1$ or less, the set of values $\{s_1,\ldots,s_D\}$ constitute a legal $D'$-out-of-$D$ sharing in some (still unknown) secret $s \in \ZZ_p$; otherwise, it is not (Lines 8-11).

\SetAlgorithmName{Sub-Protocol}{}{}
\begin{algorithm}[h!!!!]
\SetAlgoLined
\DontPrintSemicolon
    \SetKwInput{Input}{Input}
    \Input{$T_d$, $d \in [D]$, has $s_d$, a supposedly $D'$-out-of-$D$ share in some secret $s \in \ZZ_p$.}
    
\ForAll{$d \in [D]$}{

$T_d$ produces a random $r_d \in \ZZ_p$\\

$T_d$ distributes to all talliers $D'$-out-of-$D$ shares of $r_d$, denoted $r_{d,1},\ldots,r_{d,D}$\\
}

\ForAll{$d \in [D]$}{

$T_d$ computes $\hat{s}_d = s_d + \sum_{j \in [D]} r_{j,d}$\\

$T_d$ broadcasts $\hat{s}_d$\\

}

All talliers compute a polynomial $f$ of degree up to $D-1$ such that $f(d)=\hat{s}_d$, $d \in [D]$ \\

\If{$ \deg  f \leq D'-1$}{
{\bf Output} "A legal $D'$-out-of-$D$ sharing"\\
}
\Else{
{\bf Output} "An illegal $D'$-out-of-$D$ sharing"\\
}

\SetKwInOut{Output}{Output}
    \Output{The legality of the secret sharing operation.}
\caption{Verifying that the shares distributed in some scalar secret constitute a legal $D'$-out-of-$D$ sharing}
\label{verify_ss}
\end{algorithm}

\begin{lemma}
Sub-protocol \ref{verify_ss} is correct and fully preserves the voter's privacy.
\end{lemma}

\begin{proof}
Assume that $V$ is honest. Then for each entry $s=Q(m',m)$ in his ballot matrix 
there exists a polynomial $g$ of degree at most $D'-1$ such that $s_d=g(d)$, $d \in [D]$. Define
\be G:=g + \sum_{j \in [D]}g_j\,, \label{Gdef}\ee
where $g_j$, $j \in [D]$, is the polynomial of degree $D'-1$ that $T_j$ used for generating the secret shares of its random value $r_j$ (Line 3 in Sub-protocol \ref{verify_ss}). Then
\be 
\begin{aligned}
  G(d) = g(d)+\sum_{j \in [D]}g_j(d)= s_d+ \sum_{j \in [D]}r_{j,d} = \hat{s}_d\,, \qquad d \in [D]\,.\label{Geq}     
\end{aligned}    
\ee
Hence, the interpolating polynomial $f$ that the talliers compute in Line 7 coincides with $G$. As $G$ is a sum of polynomials of degree $D'-1$ at most, the validation in Lines 8-9 will pass successfully.
In that case, the talliers would learn $G(0)=g(0)+\sum_{j \in [D]} r_j$. Here, $g(0)$ is the secret entry in $V$'s ballot matrix and $r_j$ is the secret random value that $T_j$ had chosen, $j \in [D]$. Hence, $G(0)$ reveals no information on $g(0)$. Since all other coefficients in $G$ are also random numbers that are not related to $g(0)$, this validation procedure provides perfect privacy for the voter.

Assume next that $V$ had distributed illegal shares for one of his ballot matrix entries. Namely, $V$ had distributed shares $\{s_1,\ldots,s_D\}$ such that the minimum-degree polynomial $g$ that satisfies $g(d)=s_d$ for all $d \in [D]$ is of degree $t>D'-1$. In that case, the polynomial $f$ that the talliers compute in Line 7 would be
$f=g + \sum_{j \in [D]}g_j$. Since the degree of that polynomial is $t>D'-1$, they would reject the ballot (Lines 10-11). Hence, the verification procedure is correct and provides perfect privacy.
\end{proof}

When the talliers receive shares in some ballot matrix, they must first execute Sub-protocol \ref{verify_ss} for each of its entries, and only if that validation passes successfully they will proceed to perform the validation described in Section \ref{validating}.

Finally, we note that Lines 1-3 in Sub-protocol \ref{verify_ss} can be executed even before the
election period starts. Namely, once the number of registered voters, $N$, and the number of candidates, $M$, are determined, the talliers can produce $NM(M-1)/2$ random values and distribute shares of them to be used later on in masking the $M(M-1)/2$ entries in each voter's ballot matrix.  
 
\subsection{An MPC computation of the winners in the {\sc Copeland} rule}\label{seccop}
The parameter $\alpha$ in Eq. (\ref{copeland}) is always a rational number; typical settings of $\alpha$ are $0$, $1$, or $\frac{1}{2}$
~\citep{faliszewski2009llull}.
Assume that $\alpha=\frac{s}{t}$ for some integers $s$ and $t$. Then
\be t\cdot \bw(m) = t \cdot \sum_{m' \in [M] \setminus \{m\} } 1_{ P(m,m')>0 } +
s \cdot \sum_{m' \in [M] \setminus \{m\}} 1_{ P(m,m')=0 }\,. \label{tbw}
\ee
The expression in Eq. (\ref{tbw}) involves all entries in $P$ outside the diagonal. However, the talliers hold $D'$-out-of-$D$ shares, denoted $G_d(m,m')$, $d \in [D]$, in $P(m,m')$ only for entries above the diagonal, $1 \leq m < m' \leq M$ (see Lines 9-10 in Protocol \ref{algvote}).
Hence, we first
translate Eq. (\ref{tbw}) into an equivalent expression that involves only entries in $P$ above the diagonal.
Condition 3 in Theorem \ref{thm1}, together with Eq. (\ref{pdef}), imply that
$P(m',m)=-P(m,m')$.
Hence, for all $m'<m$, we can replace $1_{ P(m,m')=0 }$ with $1_{ P(m',m)=0 } $, while
$1_{ P(m,m')>0 }$ can be replaced with
$1_{ -P(m',m)>0 }$. Hence,
\begin{equation} 
\begin{aligned}
    t\cdot \bw(m) &= \\
    & t \cdot \left\{ \sum_{m' >m } 1_{ P(m,m')>0 } + \sum_{m' <m } 1_{ -P(m',m)>0 } \right\}  +  \\
    & s \cdot \left\{ \sum_{m' >m } 1_{ P(m,m')=0 } +
    \sum_{m' <m } 1_{ P(m',m)=0 } \right\}  
\end{aligned}
\label{tbw1}
\end{equation}
Eq. (\ref{tbw1}) expresses the score of candidate $C_m$, re-scaled by a factor of $t$, only by entries in $P$ above the diagonal, in which the talliers hold $D'$-out-of-$D$ secret shares. 

In view of the above, the talliers may begin by computing secret shares of the bits of positivity in the first sum on the right-hand side of Eq. (\ref{tbw1}), by using the MPC sub-protocol described in Section \ref{positivity}. As for the bits of equality to zero in the second sum on the right-hand side of Eq. (\ref{tbw1}), the talliers can compute secret shares of them using the 
MPC sub-protocol described in Section \ref{is-zero}.
As the value of $t\cdot \bw(m)$ is a linear combination of those bits, the talliers can then use the secret shares of those bits and Eq. (\ref{tbw1}) in order to get secret shares of $t\cdot \bw(m)$, for each of the candidates, $C_m$, $m \in [M]$.

Next, they perform secure comparisons among the values $t \bw(m)$, $m \in [M]$, in order to find the $K$ candidates with the highest scores.
To do that, they need to perform $M-1$ secure comparisons (as described in Section \ref{principles} and then further discussed in Section \ref{secure_comparison}) in order to find the candidate with the highest score, $M-2$ additional comparisons to find the next one, and so forth down to $M-K$ comparisons in order to find the $K$th winning candidate. Namely, the overall number of comparisons in this final stage is 
$$\sum_{m=M-K}^{M-1}m = K\cdot\left( M-\frac{K+1}{2}\right)\,.$$
The above number is bounded by $M(M-1)/2$ for all $K <M$.
Once this computational task is concluded, the talliers publish the indices of the $K$ winners (Line 11 in Protocol \ref{algvote})).

We summarize the above described computation in Sub-protocol \ref{prot-winners}, which is an implementation of Line 11 in Protocol \ref{algvote}. It assumes that the talliers
hold
$D'$-out-of-$D$ secret shares of $P(m,m')$ for all $1 \leq m<m' \leq M$. Indeed, that computation has already taken place in Lines 9-10 of Protocol \ref{algvote}. Sub-protocol \ref{prot-winners} starts with a computation of $D'$-out-of-$D$ shares of all of the positivity bits and equality to zero bits that relate to the entries above the diagonal in $P$ (Lines 1-4). Then, in Lines 5-7,
the talliers use those shares in order to obtain $D'$-out-of-$D$ shares of $t \cdot \bw(m)$ for each of the candidates, using Eq. (\ref{tbw1}); the $D'$-out-of-$D$ shares of $t \cdot \bw(m)$ are denoted $\{w_d(m): d \in [D]\}$. Finally, using the secure comparison sub-protocol, they find the $K$ winners
(Lines 8-10).

\SetAlgorithmName{Sub-Protocol}{}{}
\begin{algorithm}[htb]
\SetAlgoLined
\DontPrintSemicolon
    \SetKwInput{Input}{Input}
    \Input{$T_d$, $d \in [D]$, has $G_d(m,m')$ (a share of $P(m,m')$) for all $1 \leq m<m' \leq M$.}
\ForAll{$1 \leq m < m' \leq M$}{
The talliers apply the positivity sub-protocol to translate $\{G_d(m,m'): d \in [D]\}$ into shares $\{\sigma^+_d(m,m'): d \in [D]\}$ in $1_{P(m,m')>0}$;

The talliers apply the positivity sub-protocol to translate $\{G_d(m,m'): d \in [D]\}$ into shares $\{\sigma^-_d(m,m'): d \in [D]\}$ in $1_{-P(m,m')>0}$;

The talliers apply the equality to zero sub-protocol to translate $\{G_d(m,m'): d \in [D]\}$ into shares $\{\sigma^0_d(m,m'): d \in [D]\}$ in $1_{P(m,m')=0}$;
}
\ForAll{$d \in [D]$}{
\ForAll{$m \in [M]$}{
$T_d$ computes $w_d(m) = t \cdot \left\{ \sum_{m' >m } \sigma^+_d(m,m')
  +\sum_{m' <m } \sigma^-_d(m',m)  \right\}+$  $
s \cdot \left\{ \sum_{m' >m } \sigma^0_d(m,m') +
 \sum_{m' <m } \sigma^0_d(m',m) \right\}
$;
}
}
\ForAll{$k \in [K]:=\{1,\ldots,K\}$}{
The talliers perform $M-k$ invocations of the secure comparison sub-protocol over the $M-k+1$ candidates in $\bC$ in order to find the $k$th elected candidate;

The talliers output the candidate that was found and remove him from $\bC$;
}

\SetKwInOut{Output}{Output}
    \Output{The $K$ winning candidates from $\bC$.}
\caption{Determining the winners in {\sc Copeland} rule}
\label{prot-winners}
\end{algorithm}

\subsubsection{Privacy}\label{priv1}
The talliers hold $D'$-out-of-$D$ shares of each of the ballot matrices, $P_n$, $n \in [N]$,
as well as in the aggregated ballot matrix $P$.
Under our assumption of honest majority, and our setting of $D'= \lfloor (D+1)/2 \rfloor$, the talliers cannot recover any entry in any of the ballot matrices nor in the aggregated ballot matrix.
In computing the final election results, they input the shares that they hold into
the secure comparison sub-protocol, the positivity testing sub-protocol, or the sub-protocol that tests equality to zero (Sections \ref{secure_comparison}--\ref{is-zero}). The secure comparison sub-protocol is perfectly secure, as was shown in~\citep{NO07}. The positivity testing sub-protocol that we presented here is just an implementation of one component from the secure comparison sub-protocol, hence it is also perfectly secure. Finally, the testing of equality to zero invokes the
protocol for evaluating arithmetic functions of~\citep{DN07} and~\citep{ChidaGHIKLN18}, which was shown there to be secure.
Hence, Sub-protocol \ref{prot-winners} is perfectly secure.

\subsection{An MPC computation of the winners in the {\sc Maximin} rule}\label{secmax}
Fixing $m \in [M]$, the talliers need first to find the index $m' \neq m$ which minimizes $P(m,m')$ (where $P$ is, as before, the sum of all ballot matrices, see Eq. (\ref{pdef})). Once $m'$ is found then
$\bw(m)=P(m,m')$.
To do that (finding a minimum among $M-1$ values), the talliers need to perform $M-2$ secure comparisons. That means an overall number of $M(M-2)$ secure comparisons for the first stage in the talliers' computation of the final results (namely, the computation of the scores for all candidates under the {\sc Maximin} rule).
The second stage is just like in {\sc Copeland}, namely, finding the indices of the $K$ candidates with the highest $\bw$ scores.
As analyzed earlier, that task requires an invocation of the secure comparison sub-protocol at most $M(M-1)/2$ times. Namely, the determination of the winners in the case of {\sc Maximin} requires performing the comparison sub-protocol less than $1.5 M^2$ times.

The above described computation maintains the privacy of the voters, as argued in Section \ref{priv1}.

\subsection{A lower bound on the field's size}\label{ppp}
Here we comment on the requirements of our protocol regarding the size $p$ of the underlying field $\ZZ_p$.

The prime $p$ should be selected to be greater than the following four values:

(i) $D$, as that is the number of talliers (see Section \ref{ssharing}). 

(ii) $2N$, since the field should be large enough to hold the entries of the
    sum $P$ of all ballot matrices,
    Eq. (\ref{pdef}), and the entries of that matrix are confined to the range $[-N,N]$. 
    
(iii) $\max\{t,s\} \cdot (M-1)$, since that is the upper bound on $t \cdot \bw(m)$, see Eq. (\ref{tbw}), which is secret-shared among the talliers.

(iv) $2(M-1)$, since in validating a given ballot matrix $Q$, the talliers need to test the equality to zero of $F(Q)$, see Eq. (\ref{prodq}). As $F(Q)$ is a product of the differences $Q_m-Q_{m'}$, and each of those differences can be at most $2(M-1)$ (in {\sc Copeland}) or $M-1$ (in {\sc Maximin}), it is necessary to set $p$ to be larger than that maximal value.
    
Hence, in summary, $p$ should be selected to be larger than each of the above four values. Since $D$ (number of talliers), $M$ (number of candidates), and $s$ and $t$ (the numerator and denominator in the coefficient $\alpha$ in {\sc Copeland} rule, Eq. (\ref{copeland})), are typically much smaller than $N$, the number of voters, the essential lower bound on $p$ is $2N$.
In our evaluation, we selected $p=2^{31}-1$, which is sufficiently large for any conceivable election scenario.

\subsection{The security of the protocol}\label{overallsecurity}
An important goal of secure voting is to provide anonymity; namely, it should be impossible to connect
a ballot to the voter who cast it.
Protocol \ref{algvote} achieves that goal, and beyond. 
Indeed, each cast ballot is distributed into $D'$-out-of-$D$ random shares and then each share is designated to a unique
tallier.
Under the {\em honest majority} assumption, the voters' privacy is perfectly preserved. Namely, unless at least $D' \geq D/2$ talliers betray the trust vested in them and collude, the ballots remain secret. Therefore, not only that a ballot cannot be connected to a voter,
even the content of the ballots is never exposed, and not even the aggregated ballot matrix. The only information that anyone learns at the end of Protocol \ref{algvote}'s execution is the final election results. This is a level of privacy that exceeds mere anonymity.

A scenario in which at least half of the talliers collude
is highly improbable, and its probability decreases as $D$ increases. 
Ideally, the talliers would be parties that enjoy a high level of trust within the organization or state in which the elections take place, and whose business is based on such trust. Betraying that trust may incur devastating consequences for the talliers. Hence, even if $D$ is set to a low value such as $D=5$ or even $D=3$,
the probability of a privacy breach (namely, that $D'= \lfloor (D+1)/2 \rfloor$ talliers choose to betray the trust vested in them) would be small.

Instead of using a $D'$-out-of-$D$ threshold secret sharing, we could use an additive ``All or Nothing" secret sharing scheme, in which {\bf all} $D$ shares of all $D$ talliers are needed in order to reconstruct the shared secrets and use suitable MPC protocols for multiplication and comparison. The necessary  cryptographic machinery for such $D$-out-of-$D$ secret sharing was already developed in
\cite{BogdanovLW08}, so our protocols can be readily converted to such a setting.
That approach would result in higher security since the privacy of the voters would be jeopardized only if all $D$ talliers betray the trust vested in them and not just $D'$ out of them. However, such a scheme is not robust: if even a single tallier 
is attacked or becomes dysfunctional for whatever reason,
then all ballots would be lost. Such a risk is utterly unacceptable in voting systems. Our protocols, on the other hand, can withstand the loss of $D-D'$ talliers. 

In view of the above discussion, the tradeoff in setting the number of talliers $D$ is clear: higher values of $D$ provide higher security since more talliers would need to be corrupted in order to breach the system's security. Higher values of $D$ also provide higher robustness, since the system can withstand higher numbers of tallier failures. However, increasing $D$ has its costs: more independent and reputable talliers are needed, and the computational costs of our protocols increase, albeit modestly (see Section \ref{costanalysis}).

We would like to stress that our protocols focus on securing the computation of the election results. Needless to say that those protocols must be integrated into a comprehensive system that takes care of other aspects of voting systems.
For example, it is essential to guarantee that only registered voters can vote and that each one can vote just once. It is possible to ensure such conditions by standard means. Another requirement is the need to prevent attacks of malicious adversaries. In the context of our protocols, an adversary may eavesdrop on the communication link between some voter $V_n$ and at least $D'$ of the talliers, and intercept the messages that $V_n$ sends to them (in Protocol \ref{algvote}'s Line 7) in order to recover $V_n$'s ballot. That adversary may also replace
$V_n$'s original messages to all $D$ talliers with other messages (say, ones that carry shares of a ballot that reflects the adversary's preferences).
Such attacks can be easily thwarted by requiring each party (a voter or a tallier) to
have a certified public key,
encrypt each message that he sends out using the receiver's public key and then sign it using his own private key;
also, when receiving messages, each party must first verify them using the public key of the sender and then send a suitable
message of confirmation
to the sender. 
Namely, each message that a voter $V_n$ sends to a tallier $T_d$ in Line 7 of Protocol \ref{algvote} should be signed with $V_n$'s private key and then encrypted by $T_d$'s public key; and $T_d$ must acknowledge its receipt and verification.

\section{Order-based voting with ties}\label{ties}
So far we concentrated on the case in which the voters submit a strict ordering of the candidates, see Section \ref{formaldef}. However, order-based rules allow voters to set ties between candidates. For example, if there are six candidates, $C_1,\ldots,C_6$, a voter may identify $C_1$ as her most preferred candidate, $C_2,C_3,C_4$ as her second-tier candidates, where she is indifferent between them, and $C_5,C_6$ as her least preferred candidates. In that case, her ballot matrix would be
\be Q= \left( \begin{array}{rrrrrr} 0 & 1 & 1 & 1 & 1 & 1 \\ -1 & 0 & 0 & 0 & 1 & 1 \\
 -1 & 0 & 0 & 0 & 1 & 1 \\ -1 & 0 & 0 & 0 & 1 & 1 \\
 -1 & -1 & -1 & -1 & 0 & 0 \\  -1 & -1 & -1 & -1 & 0 & 0
\end{array} \right) \,. \label{exmp6}\ee
The aggregation of ballots and the determination of the election results is carried out as described earlier: the score $\bw(m)$ for each candidate $C_m$ is computed by Eqs. (\ref{pdef})+(\ref{copeland}), and then the winners are those with the highest scores.

The challenge that allowing ties brings about is in the characterization of legal ballot matrices and in the design of an MPC protocol for securely verifying the legality of a voter's shared ballot matrix (in the sense that it corresponds to an ordering of the candidates with possible ties). That is what we do in Sections \ref{cmties} and \ref{validatingties} below.
Our focus in those sections is the {\sc Copeland} rule. The
{\sc Maximin} rule with ties is discussed in Section \ref{maximinties}.

\subsection{Characterizing legal ballot matrices when ties are allowed}\label{cmties}

Theorem \ref{thm1ties} characterizes legal ballot matrices in the {\sc Copeland} rule when ties are allowed. It is a generalization of Theorem \ref{thm1}.

\begin{theorem}\label{thm1ties}
An $M \times M$ matrix $Q$ is a valid ballot under the {\sc Copeland} rule with ties iff it satisfies the following conditions:
\begin{enumerate}
    \item $Q(m,m') \in \{-1,0,1\}$ for all $1 \leq m < m' \leq M$;
    \item $Q(m,m)=0$ for all $m \in [M]$;
    \item $Q(m',m)+Q(m,m')=0$ for all $1 \leq m < m' \leq M$;
    \item If $Q(m',m)=0$ then $Q(m,\cdot)=Q(m',\cdot)$.
    \item For all $m \in [M]$ set $\eta_m=1$ if $Q(m',m) \neq 0$ for all $m'<m$, and $\eta_m=0$ otherwise.
    Define $Q_m:= \sum_{m' \in [M]} \eta_{m'} \cdot Q(m,m')$. Then $Q_{m'} \neq Q_m$
    for all $1 \leq m'<m \leq M$ such that $\eta_{m'}=\eta_m=1$.
\end{enumerate}
\end{theorem}

The differences between Theorem \ref{thm1} (no ties) and Theorem \ref{thm1ties} (ties allowed) are in the first condition (where in the latter theorem zero entries are allowed also in non-diagonal entries), condition 4 (that holds trivially in Theorem \ref{thm1}, as there are no zeroes outside the diagonal), and in condition 5 that generalizes condition 4 in Theorem \ref{thm1} (because when there are no ties we have $\eta_m=1$ for all $m \in [M]$).

\smallskip
Before proving the theorem, we discuss conditions 4 and 5 and exemplify their necessity.
The rational behind condition 4 is as follows. If $Q(m,m')=0$ then candidates $C_m$ and $C_{m'}$ are equivalent for the voter, and hence their relative ranking with respect to each of the other candidates (as manifested by their corresponding rows in the matrix) must be the same.
As for condition 5, let us first define the relation
$\sim$ on $[M]$ as follows: $ m \sim m'$ iff $C_m$ and $C_{m'}$ are in a tie (namely, $Q(m',m)=0$). That is obviously an equivalence relation. Let $[M] / \sim$ denote the quotient set, and
$k:= | [M] / \sim |$ be the number of equivalence classes, i.e., the number of distinct tiers in the ranking. In what follows we represent each equivalence class by the minimal index in it.
Then condition 5 considers the reduced $k \times k$ matrix over the set of representative indices, and it restates condition 4 from Theorem \ref{thm1}, namely, that all row sums in that reduced (and tie-free) matrix must be distinct.
For the sake of illustration, the ballot matrix $Q$ in Eq. (\ref{exmp6}) induces $k=3$ equivalence classes, $ [M]/\sim = \{ \{1\}, \{2,3,4\},\{5,6\}\}$. The representative indices are $1,2,5$ and the reduced $3 \times 3$ matrix is the sub-matrix of $Q$ that consists of its 1st, 2nd, and 5th rows and columns,
$$ Q|_{1,2,5}= \left( \begin{array}{rrr} 0 & 1 & 1  \\ -1 & 0 & 1 \\
 -1 & -1 &  0
\end{array} \right) \,. $$
Indeed, the row sums in that reduced matrix are distinct.

Next, we give examples to matrices that comply with conditions 1-3 of Theorem \ref{thm1ties}, but violate condition 4 or 5:
$$ Q_1 = \left( \begin{array}{rrrr} 0 & 1 & 1 & 1  \\ -1 & 0 & 0 &  1 \\
 -1 & 0 & 0 & -1 \\ -1 & -1 & 1 & 0  
\end{array} \right) ~,~~~
Q_2 = \left( \begin{array}{rrrr} 0 & 1 & -1 & -1  \\ -1 & 0 & 1 &  1 \\
 1 & -1 & 0 & 0 \\ 1 & -1 & 0 & 0  
\end{array} \right)
\,. $$
$Q_1$ violates condition 4: on one hand, $Q_1(2,3)=Q_1(3,2)=0$, which means that $C_2$ and $C_3$ are in a tie, but the corresponding rows differ, since $Q_1(2,4)=1$ while $Q_1(3,4)=-1$. However, there can be no ranking of the candidates in which $C_4$ is ranked lower than $C_2$ and higher than $C_3$, while at the same time $C_2$ and $C_3$ are in a tie.
As for $Q_2$, it does comply with condition 4: indeed, there are two candidates in a tie, $C_3$ and $C_4$, and their corresponding rows are equal. But $Q_2$ violates condition 5: indeed, while $\eta_1=\eta_2=\eta_3=1$ (and $\eta_4=0$) we have $Q_1=Q_2=Q_3=0$, in contradiction with condition 5 that requires those row sums to be distinct. That equality of the row sums is the result of the circular (and hence illegal) ranking between the candidates: $C_1>C_2>C_3=C_4>C_1$.

We now proceed to prove Theorem \ref{thm1ties}.

\begin{proof}
Assume that $Q$ is a legal {\sc Copeland} ballot matrix with ties; namely, it is a matrix that describes some given ranking over $[M]$. Then it clearly satisfies conditions 1-3. As for condition 4, it merely states that two candidates that are in a tie ($Q(m',m)=0$) must have the same relative ranking vis-a-vis all candidates, hence it clearly holds.
Finally, the row sums $Q_m$ in condition 5 are the row sums in the reduced ballot matrix over $[M]/\sim$, where $\sim$ is as defined earlier ($m \sim m'$ if $C_m$ and $C_{m'}$ are in a tie in the underlying ranking). Since there are no ties over $[M]/\sim$,
those row sums must be distinct as implied by Theorem \ref{thm1}.

Assume now that $Q$ is an $M \times M$ matrix that complies with conditions 1-5. We proceed to show that it induces a ranking over $[M]$. First, we say that $m \sim m'$, for $m,m' \in [M]$, if $Q(m',m)=0$. Clearly, that relation is reflexive (condition 2), symmetric (condition 3), and transitive (as implied by condition 4). Hence, it is an equivalence. Let us represent each equivalence class by the minimal index in it, and let $[M]/\sim$ denote the set of representative indices and $k := | [M]/\sim|$ denote their number. Then $\eta_m=1$ iff $m \in [M]/\sim$. Let $Q'$ be the reduced $k \times k$ matrix that is obtained from $Q$ by retaining only rows and columns of indices $m \in [M]/\sim$. Then $Q'$ is a matrix that complies with all four conditions in Theorem \ref{thm1} (where condition 4 in Theorem \ref{thm1} follows from condition 5 herein). Hence, $Q'$ induces a ranking (with no ties) over $[M]/\sim$.
It follows that $Q$ induces a ranking with ties over $[M]$ by adding each index in $[M] \setminus ([M]/\sim)$ to the ranking over $[M]/\sim$, next to the index in $[M]/\sim$ that is equivalent to it, with a tie between them.
The resulting ranking with ties over $[M]$ is consistent with $Q$.   
\end{proof}

\subsection{Verifying the legality of the cast ballots when ties are allowed}\label{validatingties}
Recall that if a voter's ballot matrix is $Q$ then it suffices to distribute shares only in its upper triangle $\Gamma(Q)$, Eq. (\ref{GammaDef}), see Section \ref{legal}. Herein we explain how to verify the conditions of Theorem \ref{thm1ties}. As before, conditions 2 and 3 do not need to be verified since the talliers complete $\Gamma(Q)$ into a full ballot matrix $Q$ in accordance with those anti-symmetry conditions.

The talliers can verify condition 1 if for all  $1 \leq m' < m \leq M$ they compute shares of
$$ x:= Q(m',m) \cdot (Q(m',m)+1) \cdot (Q(m',m)-1) \,,$$
using the techniques described in Section \ref{mpc}, recover $x$ and verify that $x=0$.
If all those $M(M-1)/2$ tests pass successfully then $Q$ complies with condition 1 and no other information on $Q$ is revealed.

For the verification of conditions 4 and 5 the talliers compute for each
$Q(m',m)$, $1 \leq m' < m \leq M$, shares of $\xi_{m',m} := 1_{Q(m',m)=0}$, where the latter computation is carried out as described in Section \ref{is-zero}.
After that preliminary computation, the talliers compute shares of
$$ \pi_{m',m,k}:=\xi_{m',m} \cdot \left(  Q(m',k)-Q(m,k)\right)\,, 
~~~1 \leq m' < m \leq M\,,~~ k \in [M]\,, $$
as described in Section \ref{mpc}, and then recover $\pi_{m',m,k}$.
Condition 4 is verified iff
\be \pi_{m',m,k} = 0 \,, 
~~~1 \leq m' < m \leq M\,,~~ k \in [M]\,. \label{pimkzero}\ee
Indeed, $\pi_{m',m,k} = 0 $ iff either $\xi_{m',m}=0$ or
$Q(m',k)=Q(m,k)$ for all $k \in [M]$.
Since $\xi_{m',m}=0$ iff $Q(m',m) \neq 0$, Eq. (\ref{pimkzero}) is equivalent to condition 4 in Theorem \ref{thm1ties}. Hence, a matrix $Q$ passes that test iff it complies with condition 4, and if it does, no further information on it is leaked as a result of that verification.
Note that the talliers do not hold shares of $Q(k,m)$ whenever $k \geq m$. Hence, in such cases the talliers rely on conditions 2 and 3 and make the following substitutions: when $k=m$ they set $Q(k,m)=0$, and when $k>m$ they set $Q(k,m) = -Q(m,k)$.

Finally, we turn to the verification of condition 5. Here, the talliers compute for all $m \in [M]$ shares of
\be \eta_m := \prod_{m'=1}^{m-1} \left( 1 - \xi_{m',m}\right) \,. \label{etadef}\ee
It can be easily verified that $\eta_m = 1$ if $Q(m',m) \neq 0$ for all $1 \leq m' <m$, while $\eta_m=0$ otherwise.
Hence, the indicator variables $\eta_m$ as defined above equal the indicator variables $\eta_m$ that are defined in Theorem \ref{thm1ties}'s condition 5. The computation in Eq. (\ref{etadef}) is carried out by the method described in Section \ref{mpc}. 
Next, the talliers compute shares of
\be Q_m = \sum_{m' \in [M]} \eta_{m'} \cdot Q(m,m') \label{Qities}\,,\ee
where, as before, we recall that $Q(m,m)=0$ and when $m'>m$ we rely on the equality $Q(m',m)=-Q(m,m')$. Condition 5 requires that for all $1 \leq m'<m\leq M$, whenever $\eta_m \eta_{m'} \neq 0$ then also $Q_{m'}-Q_m \neq 0$.
That condition can be rephrased as follows:
\be \gamma_{m',m} :=(1-\eta_m\eta_{m'} ) + \eta_m\eta_{m'} \cdot ( Q_{m'}-Q_m) \neq 0 \,,~~~ 
 1 \leq m'<m\leq M \,.  \label{gammamm}\ee
Indeed, if $\eta_m=0$ or $\eta_{m'} =0$, a case in which condition 5 is void, then the first addend in Eq. (\ref{gammamm}) equals 1 while the second addend is zero, so $\gamma_{m',m}=1$. If, on the other hand,
$\eta_m\eta_{m'}=1$, then condition 5 requires the difference 
$Q_{m'}-Q_m$ to be nonzero, and since in that case $\gamma_{m',m}=Q_{m'}-Q_m$, the inequalities in Eq. (\ref{gammamm}) are indeed equivalent to condition 5.

The talliers could compute shares of $\gamma_{m',m}$, using the techniques described in Section \ref{mpc}, and then recover $\gamma_{m',m}$ and check that it is nonzero. Alas, the value of $\gamma_{m',m}$ will reveal information on the ballot matrix. To refrain from such leakage of information, the talliers produce offline (even prior to the election period) a large pool of shared nonzero random values from the underlying field $\ZZ_p$. Then, for each value $\gamma_{m',m}$ they would pull a fresh nonzero random value $r$ from that pool and then recover $r \gamma_{m',m}$, using the shares that they had computed in $\gamma_{m',m}$, the shares that they had computed offline in $r$, and the multiplication procedure from Section \ref{mpc}.
Since $r \gamma_{m',m} \neq 0$ iff $\gamma_{m',m} \neq 0$, the talliers can verify condition 5 without learning any further information on $Q$. 

We summarize the verification procedure in Sub-protocol \ref{verify}. In Lines 1-5 the talliers complete $\Gamma(Q)$ into a full ballot matrix $Q$, based on $Q$'s asymmetry. 
In Lines 6-10 they verify condition 1; in Lines 11-17 they verify condition 4; and in Lines 18-29 they verify condition 5. In any event of detecting that the ballot is illegal, the sub-protocol outputs a message about that illegality and aborts (Lines 10, 17, 29). If the sub-protocol reaches the end, it outputs
a message about the legality of the input ballot matrix and then terminates (Line 30). Note that the computation in Line 25 is independent of the input and can be carried out offline, even before the election period had started.

\subsection{The {\sc Maximin} rule with ties}\label{maximinties}
As discussed in Section \ref{formaldef}, the ballot matrix in {\sc Maximin} consists of 0,1 entries, where the $(m,m')$-th entry in the matrix equals 1 iff $C_m$ is ranked strictly higher than $C_{m'}$ in the voter's ranking. Instead of formalizing a characterization and designing a corresponding verification MPC protocol for such matrices, we propose that also in
{\sc Maximin} rule, just like in {\sc Copeland}, the voters will submit a matrix of $\{-1,0,1\}$-entries that fully spells out their ranking. The talliers can then verify the legality of that matrix by running Sub-protocol \ref{verify}. After the ballot matrix is verified to be legal, the talliers may easily translate it to a {\sc Maximin} ballot matrix of $\{0,1\}$-entries, as described in Section \ref{formaldef}, by running the following computation on each of the matrix entries, $x = Q(m,m')$, $1 \leq m \neq m' \leq M$:
\be x \leftarrow x + 1_{x+1=0} \,. \label{xupd}\ee
The computation in Eq. (\ref{xupd}) involves an invocation of the equality to zero protocol of Section \ref{is-zero}.
It is easy to see that the computation in Eq. (\ref{xupd}) 
substitutes (shares of) $x=-1$ with (shares of) $x=0$, but leaves (the shares of) $x$ unchanged if $x \neq -1$. After that update, the ballot matrix (that was already verified to be a legal {\sc Copeland} ballot matrix) will be a legal {\sc Maximin} matrix, consisting only of $0,1$ entries, and the aggregation of all ballots and the computation of the election results will be as described in 
Section \ref{secmax}.

\SetAlgorithmName{Sub-Protocol}{}{}
\begin{algorithm}[h!!!!]
\SetAlgoLined
\DontPrintSemicolon
    \SetKwInput{Input}{Input}
    \Input{$T_d$, $d \in [D]$, has $Q_d(m',m)$ (a share of $Q(m',m)$, where $Q$ is a ballot matrix of some voter) for all $1 \leq m'<m \leq M$.}
\ForAll{$d \in [D]$}{

\ForAll{$1 \leq m \leq M$}{

$T_d$ sets $Q_d(m,m) \leftarrow 0$ \\

\ForAll{$m<m'\leq M$}{

$T_d$ sets $Q_d(m',m) \leftarrow -Q_d(m,m')$ \\

}
}
}

\ForAll{$1 \leq m' < m \leq M$}{
Compute shares of $x:=Q(m',m) \cdot (Q(m',m)+1) \cdot (Q(m',m)-1) $ \\

Recover $x$ \\

\If{$x \neq 0$}
{
{\bf Output} "Illegal ballot" and {\bf Abort}
}
}

\ForAll{$1 \leq m' < m \leq M$}{

Compute shares of $\xi_{m',m}:=1_{Q(m',m)=0}$ \\

\ForAll{$1 \leq k \leq M$}{

Compute shares of $ \pi_{m',m,k}:=\xi_{m',m} \cdot \left(  Q(m',k)-Q(m,k)\right)$ \\

Recover $ \pi_{m',m,k}$ \\

\If{$\pi_{m',m,k} \neq 0$}
{
{\bf Output} "Illegal ballot" and {\bf Abort}
}
}
}

\ForAll{$1 \leq m \leq M$}{

Compute shares of $\eta_m := \prod_{m'=1}^{m-1} ( 1 - \xi_{m',m}) $ \\

}

\ForAll{$1 \leq m \leq M$}{

Compute shares of $Q_m:= \sum_{m'=1}^M \eta_{m'} \cdot Q(m,m')$ \\

}

\ForAll{$1 \leq m'<m \leq M$}{

Compute shares of $\eta_m\eta_{m'}$ \\

Compute shares of $\gamma_{m',m} :=(1-\eta_m\eta_{m'} ) + \eta_m\eta_{m'} \cdot ( Q_{m'}-Q_m)$ \\

Compute shares of a nonzero random $r \in \ZZ_p$ \\

Compute shares of $x:=r \cdot \gamma_{m',m}$ \\

Recover $x$ \\

\If{$x= 0$}
{
{\bf Output} "The ballot is illegal" and {\bf Abort}
}

}

{\bf Output} "The ballot is legal" and {\bf Terminate}

\SetKwInOut{Output}{Output}
    \Output{The legality of the input ballot matrix $Q$.}
\caption{Verifying the legality of a ballot matrix in {\sc Copeland} rule with ties}
\label{verify}
\end{algorithm}

\section{An implementation of a secure order-based voting system}\label{demo}
Our goal herein is to establish the practicality of our protocol and demonstrate it in action. 
To that end, we had implemented a demo of a voting system based on the protocols that we presented here for the {\sc Copeland} and {\sc Maximin} voting rules. 
The full source code of the demo is open source and is available on GitHub (\url{https://github.com/arthurzam/SecureVoting}). The demo is available at \url{https://securevoting.ddns.net/}. Interested readers can use it to define a new election (by setting all the necessary parameters) and invite voters to cast their votes and then compute the final election results.

The system is implemented in {\sf Python}. We chose that programming language for ease of development and best cross-platform support. In addition, it makes the code easier to read, verify, and modify as required. 

In our system there are users (clients) and talliers (servers). A user can be either an election manager, namely a party that initiates a new election and defines its characteristics, or a voter in an election that was initiated and is managed by another user. As for the talliers, even though our code can be executed with any number of talliers, we focus in our demo on the case of three talliers.

The demo uses {\sf docker} and {\sf docker-compose} to easily deploy on servers. 
The deployment includes a nearly-static Web server (which supplies the user interface), a {\sf PostgreSQL} database (holding details on all currently defined elections, either future ones, on-going elections, or complete elections, and on all voters), and the tallier module. A high-level block diagram of the demo is given in Figure \ref{fig: block-diagram}. For the sake of simplicity, the figure illustrates only two talliers; the structure of all talliers, and the communication between each tallier and the client and between each pair of talliers is as shown in the figure for Talliers 1 and 2.

\begin{figure*}[ht!]
    \begin{center}
    \begin{tikzpicture}[y=1cm, x=1cm, yscale=0.75,xscale=0.75, every node/.append style={scale=0.75}, inner sep=0pt, outer sep=0pt]
  \path[draw=black] (6.8792, 9.0223) rectangle (14.6844, 5.0535);
    \node[text=black,anchor=south] (text6614) at (10.795, 8.5196){\textbf{Tallier 1}};
  \path[draw=black,rounded corners=0.3969cm] (7.4083, 8.2285) rectangle (10.5833, 5.5827);
    \node[text=black,anchor=south] (text8851) at (8.9958, 7.7788){Docker};
  \path[draw=black,miter limit=10.0] (7.9375, 7.3025) -- (3.079, 7.3025);
  \path[draw=black,fill=black,miter limit=10.0] (2.9401, 7.3025) -- (3.1253, 7.3951) -- (3.079, 7.3025) -- (3.1253, 7.2099) -- cycle;

    \node[text=black,anchor=south] (text5161) at (5.1065, 7.3819){Election Info};
  \path[draw=black] (7.9375, 7.6994) rectangle (10.0542, 6.1119);
    \node[text=black,anchor=south,align=center] (text3020) at (8.9958, 6.45){WebSocket\\Server};
  \path[draw=black,rounded corners=0.3969cm] (11.1125, 8.2285) rectangle (14.3875, 5.5827);
    \node[text=black,anchor=south] (text4557) at (12.7, 7.7788){Docker};

  \path[draw=black] (11.6417, 7.6994) rectangle (13.9583, 6.1119);
    \node[text=black,anchor=south,align=center] (text1145) at (12.8, 6.45){PostgreSQL\\DB};

  \path[draw=black,miter limit=10.0] (10.2227, 6.9056) -- (11.4731, 6.9056);
  \path[draw=black,fill=black,miter limit=10.0] (10.0838, 6.9056) -- (10.269, 6.9982) -- (10.2227, 6.9056) -- (10.269, 6.813) -- cycle;

  \path[draw=black,fill=black,miter limit=10.0] (11.612, 6.9056) -- (11.4268, 6.813) -- (11.4731, 6.9056) -- (11.4268, 6.9982) -- cycle;

  \path[draw=black] (6.8792, 3.9952) rectangle (14.6844, 0.0265);
  \node[text=black,anchor=south] (text5925) at (10.795, 0.1058){\textbf{Tallier 2}};

  \path[draw=black,rounded corners=0.3969cm] (7.4083, 3.466) rectangle (10.5833, 0.8202);
  \node[text=black,anchor=south] (text3823) at (8.9958, 0.9){Docker};

  \path[draw=black] (7.9375, 2.9369) rectangle (10.0542, 1.3494);
  \node[text=black,anchor=south,align=center] (text350) at (8.9958, 1.7){WebSocket\\Server};

  \path[draw=black,rounded corners=0.3969cm] (11.1125, 3.466) rectangle (14.3875, 0.8202);
  \node[text=black,anchor=south] (text2302) at (12.7, 0.9){Docker};

  \path[draw=black] (11.6417, 2.9369) rectangle (13.9583, 1.3494);
  \node[text=black,anchor=south,align=center] (text5028) at (12.8, 1.7){PostgreSQL\\DB};

  \path[draw=black,miter limit=10.0] (10.2227, 2.1431) -- (11.4731, 2.1431);

  \path[draw=black,fill=black,miter limit=10.0] (10.0838, 2.1431) -- (10.269, 2.2357) -- (10.2227, 2.1431) -- (10.269, 2.0505) -- cycle;
  \path[draw=black,fill=black,miter limit=10.0] (11.612, 2.1431) -- (11.4268, 2.0505) -- (11.4731, 2.1431) -- (11.4268, 2.2357) -- cycle;
  \path[draw=black,miter limit=10.0] (8.9958, 3.1054) -- (8.9958, 5.9433);

  \path[draw=black,fill=black,miter limit=10.0] (8.9958, 2.9665) -- (8.9032, 3.1517) -- (8.9958, 3.1054) -- (9.0884, 3.1517) -- cycle;

  \path[draw=black,fill=black,miter limit=10.0] (8.9958, 6.0822) -- (9.0884, 5.897) -- (8.9958, 5.9433) -- (8.9032, 5.897) -- cycle;

    \node[text=black,anchor=south west] (text6122) at (9.234, 4.4185){TLS 1.3};

  \path[draw=black] (6.8792, 13.2556) rectangle (11.1125, 9.5515);
    \node[text=black,anchor=south] (text4651) at (8.9958, 12.7529){\textbf{Web Server \& UI}};

  \path[draw=black,rounded corners=0.3969cm] (7.3422, 12.4619) rectangle (10.5172, 9.816);
    \node[text=black,anchor=south] (text1005) at (8.9429, 12.0121){Docker};

  \path[draw=black,miter limit=10.0,rounded corners=0.2cm] (7.8714, 11.139) -- (1.4552, 11.139) -- (1.4552, 7.8679);
  \path[draw=black,fill=black,miter limit=10.0] (1.4552, 7.729) -- (1.3626, 7.9142) -- (1.4552, 7.8679) -- (1.5478, 7.9142) -- cycle;
    \node[text=black,anchor=south] (text9573) at (4.1275, 11.2183){HTTPS};

  \path[draw=black] (7.8714, 11.9327) rectangle (9.988, 10.3452);
    \node[text=black,anchor=south,align=center] (text8914) at (8.9429, 10.4){Static\\Content\\Supplier};
  \path[draw=black,miter limit=10.0] (2.9104, 6.5088) -- (7.769, 6.5088);

  \path[draw=black,fill=black,miter limit=10.0] (7.9079, 6.5088) -- (7.7227, 6.4161) -- (7.769, 6.5088) -- (7.7227, 6.6014) -- cycle;
  \node[text=black,anchor=south] (text8303) at (4.9477, 5.9796){Ballot};

  \path[draw=black,miter limit=10.0,rounded corners=0.2cm] (7.9375, 2.54) -- (2.1828, 2.54) -- (2.1828, 5.9433);
  \path[draw=black,fill=black,miter limit=10.0] (2.1828, 6.0822) -- (2.2754, 5.897) -- (2.1828, 5.9433) -- (2.0902, 5.897) -- cycle;
  \node[text=black,anchor=south] (text2356) at (4.9213, 2.5929){Election Info};

  \path[draw=black,miter limit=10.0,rounded corners=0.2cm] (1.4552, 6.1119) -- (1.4552, 1.7463) -- (7.769, 1.7463);
  \path[draw=black,fill=black,miter limit=10.0] (7.9079, 1.7463) -- (7.7227, 1.6536) -- (7.769, 1.7463) -- (7.7227, 1.8389) -- cycle;
  \node[text=black,anchor=south] (text3610) at (5.1594, 1.27){Ballot};

  \path[draw=black] (0.0, 7.6994) rectangle (2.9104, 6.1119);
  \node[text=black,anchor=south,align=center] (text473) at (1.4552, 6.4){\textbf{Client}\\(inside browser)};

\end{tikzpicture}
    \caption{A high-level block diagram of the demo's system.}
    \label{fig: block-diagram}
    \end{center}
\end{figure*}

\subsection{Web server and user interface}
The Web server is implemented using the {\sf Flask} and {\sf Jinja2} libraries, which render nearly static Web pages that serve as the user interface for the demo. Importantly, the Web server does not store any information or connect to the {\sf PostgreSQL} database,
thus ensuring a stateless and lightweight design. For real-world deployment, it is recommended to host the Web server on a separate machine in order to mitigate potential security risks.

The client-side code is developed using {\sf HTML} and {\sf JavaScript}. It is responsible for presenting to the voter the election form (through which the voter can conveniently input her or his preferred ranking). After the voter had submitted her vote, the client generates the corresponding ballot matrix, computes the corresponding secret shares according to the protocol, and distributes them to the talliers.

\subsection{The tallier module}
The tallier module is implemented using the {\sf websocket} and {\sf asyncpg} libraries. Communication between the user and any given tallier is done by an encrypted TLS 1.3 {\sf WebSocket} channel, for ensuring secure data transmission. Users must register with all talliers, create an election, and submit their votes through this interface.

The MPC computations are implemented using {\sf async} functions, which align with the natural structure of the protocol. This design closely mirrors implementations in referenced research papers, making the codebase easier to verify and extend.

To facilitate future development, a robust unit-testing setup has been implemented. This setup enables developers to verify the correctness of the implementation after modifications, and to streamline the development process.

\subsection{User's usage flow}
Upon entering the demo, new users are directed to
a registration page, in which they fill out a registration form by providing their email address and full name.
Consequently, a password will be sent to their email inbox, ensuring secure account setup.
After registration, users return to the main page of the demo in which they can log in. 
The top bar on that page displays the communication status with each of the talliers, for real-time feedback on system connectivity.

Once logged in, users are directed to their dashboard, which displays all the elections they are managing or participating in as voters.
(For new users, the list is empty.) From the dashboard, users can start a new election (for which they will be the managers), vote in an ongoing election (in which they are voters), and view the results of completed elections.

If a user wishes to start a new election, she must fill in an election creation form, see Figure \ref{fig: demo-create-election}. In that form she will configure various election parameters, including: election title,
the voting rule (currently we implement two order-based rules, {\sc Copeland} and {\sc Maximin}), the names of all candidates, the desired number of winning candidates, and the emails of all eligible voters. The list of voters' emails can be provided either manually or by a CSV file in which each row holds an email address of a single voter.

\begin{figure}[h!!]
    \centering
    \includegraphics[width=0.5\linewidth]{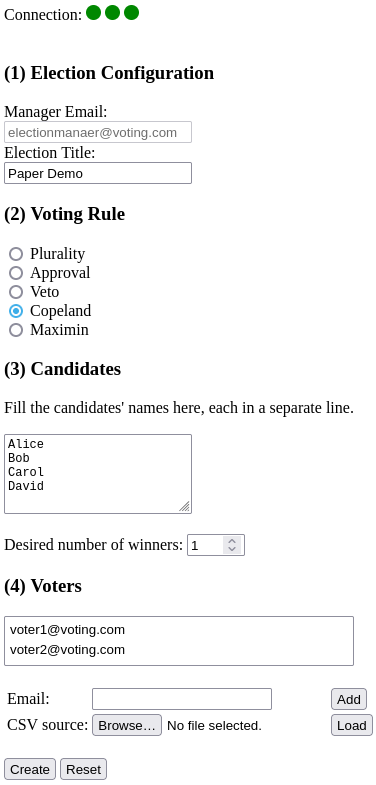}
    \caption{Election creation form}
    \label{fig: demo-create-election}
\end{figure}

Upon the successful creation of a new election, the election manager is redirected back to the dashboard. From there, the manager can start the voting session for the election. By doing so, an email with election details and a link for submitting the vote will be sent to all eligible voters.

Each voter, upon accessing the voting link, fills in her preferences, by determining her ranking of the candidates (where ties are allowed), see Figure \ref{fig: demo-vote}. The client translates that ranking to a corresponding ballot matrix and then generates secret shares of its (above diagonal) entries and distributes them to the talliers.
The talliers perform the proper validation checks, and if the ballot is verified they inform the voter that they had successfully received her ballot.

The election manager ends the election once the pre-determined election period had elapsed. At that point the talliers finalize the process by securely computing the winners and sending to all voters and to the manager an email with the election results.

\begin{figure}
    \centering
    \includegraphics[width=0.5\linewidth]{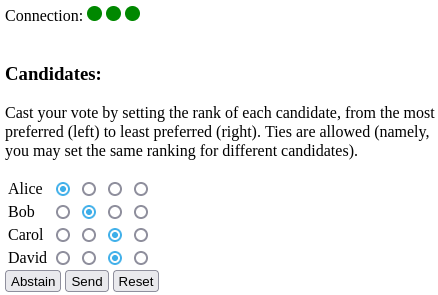}
    \caption{Voting page}
    \label{fig: demo-vote}
\end{figure}

\subsection{Runtime evaluation}\label{costanalysis}
We have implemented our protocol over the field $\ZZ_p$ where $p=2^{31}-1$. As explained in Section \ref{ppp}, such a setting is sufficient for any conceivable election scenario.

The two parameters that affect the protocol's runtime are $D$, the number of talliers, and $M$, the number of candidates. Based on our discussion in Section \ref{overallsecurity} we set the value of $D$ to $D \in \{3,5,7,9\}$.  
The value of $M$ was selected from the range $M \in \{3,5,10,15,20\}$. This range is representative of real-world elections, which rarely feature higher numbers of candidates or parties. In most national elections the number of competing entities remains modest—for example, parliamentary elections in many European countries typically involve fewer than 20 parties. Moreover, it is cognitively unrealistic to expect voters to meaningfully rank a large number of candidates. Empirical evidence suggests that individuals can comfortably evaluate and rank only a small number of alternatives (see, e.g.,~\citep{lau2001advantages,miller1956magical}). Therefore, our choice of $M$ values reflects both realistic election settings and voter capabilities.

As for the number of voters, $N$, it affects only the runtime for validating incoming ballots. The time for computing the final election results, on the other hand, is independent of $N$. The runtimes reported in our experiments correspond to the cost of validating either a single ballot or a fixed-size batch of $b=64$ ballots. Therefore, the total validation runtime for an election with $N$ voters is obtained by multiplying the per-ballot cost by $N$, or multiplying the $b$-batch cost by $N/b$.


For each of the two voting rules, {\sc Copeland} and {\sc Maximin}, we measured the average runtime for validating the cast ballots, as well as measuring the time to compute the final election results at the end of the election period, and present their dependence on $D$ and $M$. Indeed, the validation of a single ballot or a batch of ballots does not depend on $N$, while the computation of the final election results depends only on the dimension $M \times M$ of the aggregated ballot matrix $P = \sum_{n=1}^N P_n$, Eq. (\ref{pdef}), but not on the number $N$ of addends in the sum that defines it.

All experiments were executed on three servers that resided on the same LAN. Each of the servers had two AMD EPYC 7543 32-core processors, 128 GB RAM, and a "Gentoo Linux" operation system. 

In the first experiment we measured the time for validating a single ballot in each of the two rules. 
As our implementation allows ties, the presented runtimes are for Sub-protocol \ref{verify}. Recall that the same protocol is used for both rules $-$ {\sc Copeland} and {\sc Maximin}. The results are shown in Figure \ref{fig: validsingleballot}. As can be seen, larger values of $D$ entail computing more secret shares and sending more messages and, consequently, the overall runtime increases with $D$. As for $M$, the quadratic dependence of the runtime on $M$ is clearly shown in the figure.

\begin{figure*}[h!!]
    \begin{center}
        \begin{tikzpicture}
        \begin{axis}[
            xlabel={$M$}, ylabel={Time ($ms$)},
            xmin=2, xmax=21,  ymin=0, ymax=270,
            xtick={3,5,10,15,20},
            ytick={0,50,100,150,200,250},
            legend pos=north west, ymajorgrids=true, grid style=dashed,
        ]        
        \addplot[solid,mark=*] coordinates {(3,11)(5,13)(10,28)(15,53)(20,69)};
            \addlegendentry{$D=3$}
        \addplot[color=red,dashed,mark=*] coordinates {(3,13)(5,17)(10,34)(15,67)(20,118)};
            \addlegendentry{$D=5$}
        \addplot[color=blue,dotted,mark=*] coordinates {(3,19)(5,25)(10,51)(15,102)(20,181)};
            \addlegendentry{$D=7$}
        \addplot[color=brown,dashdotted, mark=*] coordinates {(3,26)(5,33)(10,71)(15,140)(20,261)};
            \addlegendentry{$D=9$}
        \end{axis}
        \end{tikzpicture}
    \caption{Runtimes (milliseconds) for validating a single ballot in each of the two rules, as a function of the number of candidates, $M$, and number of talliers, $D$.}
    \label{fig: validsingleballot}
    \end{center}
\end{figure*}

In the second experiment we measured the time for validating a batch of ballots, where we used batches of size 64. By validating several ballots in parallel, we can compute concurrently independent multiplications and thus reduce the overall runtime. In Figure \ref{fig: validbatch} we show the ratio between the time to validate a batch of 64 ballots and the time to validate 64 single ballots. All shown values are smaller than 1 (namely, batching always saves time), where the more significant improvements (i.e., smaller ratios) were obtained for higher values of $M$ and $D$. For example, if we concentrate on the highest parameters in our experiments, $M=20$ and $D=9$, then the average time to validate a single ballot, when using a 64-batch validation strategy, is roughly 50 msec, as implied by the values in Figures \ref{fig: validsingleballot} and \ref{fig: validbatch}. That means that within a day it is possible to validate over 1,700,000 ballots. Clearly, those numbers can be further improved by using larger batches.

\begin{figure*}[h!!]
    \begin{center}
        \begin{tikzpicture}
        \begin{axis}[
            xlabel={$M$}, ylabel={Improvement factor},
            xmin=2, xmax=21,  ymin=0, ymax=1,
            xtick={3,5,10,15,20},
            ytick={0,0.2,0.4,0.6,0.8},
            legend pos=north east, ymajorgrids=true, grid style=dashed,
        ]        
        \addplot[solid,mark=*] coordinates {(3,0.91)(5,0.83)(10,0.64)(15,0.58)(20,0.48)};
            \addlegendentry{$D=3$}
        \addplot[color=red,dashed,mark=*] coordinates {(3,0.78)(5,0.66)(10,0.4)(15,0.35)(20,0.31)};
            \addlegendentry{$D=5$}
        \addplot[color=blue,dotted,mark=*] coordinates {(3,0.56)(5,0.39)(10,0.29)(15,0.23)(20,0.21)};
            \addlegendentry{$D=7$}
        \addplot[color=brown,dashdotted, mark=*] coordinates {(3,0.31)(5,0.31)(10,0.24)(15,0.20)(20,0.18)};
            \addlegendentry{$D=9$}
        \end{axis}
        \end{tikzpicture}
    \caption{The improvement factor in runtime when validating batches of 64 ballots.}
    \label{fig: validbatch}
    \end{center}
\end{figure*}

By choosing larger-size batches we could reduce the average runtime for validating any single ballot even further. However, while selecting larger batch sizes translates to reduced computation times, it also translates to increased response times, since the validation of any single ballot will be delayed until a sufficient number of ballots is received in order to form a batch of the desired size. If we wish to issue to the voter an immediate response regarding the validity of the ballot that he had submitted (so that in case the received ballot was not validated the voter would be asked to submit a new ballot), then the system should perform single ballot validation or set a short time window and then validate in parallel the batch of all ballots that were received during that window. 

In the third and fourth experiments we measured the time to compute the winners in {\sc Copeland} and {\sc Maximin} rules; those runtimes are shown in Figures \ref{fig: winnersCopeland} and  \ref{fig: winnersMaximin} respectively. 

\begin{figure*}[h!!]
    \begin{center}
        \begin{tikzpicture}
        \begin{axis}[
            xlabel={$M$}, ylabel={Time ($s$)},
            xmin=2, xmax=21,  ymin=0, ymax=11,
            xtick={3,5,10,15,20},
            ytick={0,2,4,6,8,10},
            legend pos=north west, ymajorgrids=true, grid style=dashed,
        ]        
        \addplot[solid,mark=*] coordinates {(3,0.04)(5,0.119)(10,0.488)(15,1.254)(20,2.410)};
            \addlegendentry{$D=3$}
        \addplot[color=red,dashed,mark=*] coordinates {(3,0.065)(5,0.188)(10,0.849)(15,2.331)(20,4.890)};
            \addlegendentry{$D=5$}
        \addplot[color=blue,dotted,mark=*] coordinates {(3,0.098)(5,0.302)(10,1.413)(15,3.539)(20,7.657)};
            \addlegendentry{$D=7$}
        \addplot[color=brown,dashdotted, mark=*] coordinates {(3,0.133)(5,0.361)(10,1.720)(15,5.033)(20,10.129)};
            \addlegendentry{$D=9$}
        \end{axis}
        \end{tikzpicture}
    \caption{Runtimes (seconds) for computing election results in {\sc Copeland}.}
    \label{fig: winnersCopeland}
    \end{center}
\end{figure*}

\begin{figure*}[h!!]
    \begin{center}
        \begin{tikzpicture}
        \begin{axis}[
            xlabel={$M$}, ylabel={Time ($s$)},
            xmin=2, xmax=21,  ymin=0, ymax=3,
            xtick={3,5,10,15,20},
            ytick={0,0.5,1,1.5,2,2.5,3},
            legend pos=north west, ymajorgrids=true, grid style=dashed,
        ]        
        \addplot[solid,mark=*] coordinates {(3,0.013)(5,0.04)(10,0.168)(15,0.386)(20,0.719)};
            \addlegendentry{$D=3$}
        \addplot[color=red,dashed,mark=*] coordinates {(3,0.022)(5,0.072)(10,0.301)(15,0.770)(20,1.396)};
            \addlegendentry{$D=5$}
        \addplot[color=blue,dotted,mark=*] coordinates {(3,0.032)(5,0.105)(10,0.479)(15,1.054)(20,2.033)};
            \addlegendentry{$D=7$}
        \addplot[color=brown,dashdotted, mark=*] coordinates {(3,0.037)(5,0.121)(10,0.547)(15,1.391)(20,2.728)};
            \addlegendentry{$D=9$}
        \end{axis}
        \end{tikzpicture}
    \caption{Runtimes (seconds) for computing election results in {\sc Maximin}.}
    \label{fig: winnersMaximin}
    \end{center}
\end{figure*}

\section{Related work}\label{related}
Secure e-voting can be approached using various cryptographic techniques. 
The earliest suggestion is that of \citet{chaum1981untraceable}, who suggested using a mix network (mixnet). The idea is to treat the ballots as ciphertexts. 
Voters encrypt  their ballots
and agents collect  and shuffle  these messages and thus anonymity of the ballots is preserved. Other studies followed and improved this model, e.g. \citet{sako1995receipt, adida2008helios, boneh2002almost,jakobsson2002making,lee2003providing,neff2001verifiable}.
However, while such systems preserve anonymity, the talliers are exposed to the actual ballots. The mere anonymity of the ballots might not provide sufficient security and this may encourage voters to abstain or vote untruthfully 
\citep{dery2021fear}. 

One of the approaches towards achieving tally-hiding privacy, and not just anonymity, is by employing homomorphic encryption, which allows performing computations on encrypted values without decrypting them first. The most common ciphers of that class are additively homomorphic, in the sense that the product of several ciphertexts is the encryption of the sum of the corresponding plaintexts.
Such encryptions are suitable for secure voting, as was first suggested by \citet{Benaloh86a}.
The main idea is to encrypt the ballots using a public-key homomorphic cipher.
An agent aggregates the encrypted ballots and then sends an aggregated encrypted value to the tallier. The tallier decrypts the received ciphertexts and recovers the aggregation of the ballots, but is never exposed to the ballots themselves. Secure voting protocols that are based on homomorphic encryption were presented in e.g. \citep{cramer1997secure, damgaard2010generalization, fan2020hse, hevia2004electronic,  priya2018blockchain,rezaeibagha2019provably,yang2018secure}. 

While most studies on secure voting offered protocols for securing the voting process, some studies considered the question of private execution of the computation that the underlying voting rule dictates.
We begin our survey with works that considered score-based voting rules.
\citet{canard2018practical} considered
the {\sc Majority Judgment} (MJ) voting rule~\citep{balinski2007theory}. 
They first translated the complex control flow and branching instructions that the MJ rule entails into a branchless algorithm; then they devised a privacy-preserving implementation of it using homomorphic encryption, distributed decryption schemes, distributed evaluation of Boolean gates, and distributed comparisons.
\citet{NairBK15} suggested to use secret sharing for the
tallying process in {\sc Plurality} voting \citep{brandt2016handbook}. Their protocol provides anonymity but does not provide
perfect secrecy as it reveals the final aggregated score of each candidate. In addition, their protocol is vulnerable to cheating attacks, as it does not include means for detecting illegal votes.
\citet{kusters2020ordinos} introduced a secure end-to-end verifiable tally-hiding e-voting system, called Ordinos, that implements the {\sc Plurality} rule and outputs the $K$ candidates that received the highest number of votes, or those with number votes that is greater than some threshold.
\citet{dery2021fear} offered a solution based on MPC in order to securely determine the winners in elections governed by score-based voting rules, including
{\sc Plurality}, {\sc Approval}, {\sc Veto}, {\sc Range} and {\sc Borda}. Their protocols offer perfect privacy and very attractive runtimes.

Recently, few researchers have begun looking at order-based voting rules.
\citet{haines2020verifiable} proposed a solution for the order-based {\sc Schulze} rule \citep{schulze2011new}. Their solution
does not preserve the privacy of voters who are indifferent between some pairs of candidates. In addition, their solution is not scalable to large election campaigns, as they report a runtime of 25 hours for an election with 10,000 voters.
\citet{hertel2021extending} proposed solutions for {\sc Copeland}, {\sc Maximin} and {\sc Schulze} voting rules. The evaluation of
the {\sc Schulze} method took 135 minutes for 5 candidates and 9 days, 10 hours, and 27 minutes for 20 candidates.
Finally, \citet{cortier2022toolbox} considered the {\sc Single Transferable Vote} (STV) rule, which is a multi-stage rule, as well as {\sc Schulze} rule. Even though their method is much more efficient than the one in \citet{hertel2021extending} for the {\sc Schulze} rule, it is still not scalable to large election settings, as it took 8 hours and 50 minutes for $N=1024$ voters and $M=20$ candidates.

Our study is the first one that proposes protocols for order-based rules that are both fully private and lightweight so they offer a feasible and fast solution even for democracies of millions of voters, as our experiments indicate (see Section \ref{costanalysis}). The overwhelming advantage, in comparison to the above-mentioned recent works, is achieved mainly by our novel idea of distributed tallying.

\section{Conclusion}\label{conc}
Here we summarize our study (Section \ref{summary}), describe a real-life implementation of the proposed protocol and system (Section \ref{gentoo}), and outline future research directions (Section \ref{future}).

\subsection{Summary}\label{summary}
In this study we presented a protocol for the secure computation of order-based voting rules. Securing the voting process is an essential step toward a fully online voting process.
A fundamental assumption in all secure voting systems that rely on fully trusted talliers (that is, talliers who receive the actual ballots from the voters) is that the talliers do not misuse the ballot information and that they keep it secret.
In contrast, our protocol significantly reduces the trust vested in the talliers, as it denies the talliers access to the actual ballots and utilizes MPC techniques in order to compute the desired output without allowing any party access to the inputs (the private ballots). Even in scenarios where some (a minority) of the talliers betray that trust, privacy is ensured. Such a reduction of trust in the talliers is essential to increase the confidence of the voters in the voting system so that they would be further motivated to exercise their right to vote and, moreover, vote according to their true preferences, without fearing that their private vote will be disclosed to anyone.

Our protocol offers perfect ballot secrecy: 
the protocol outputs the identity of the winning candidates, but the voters as well as the talliers remain oblivious to any other related information, such as the actual ballots or any other value that is computed during the tallying process (e.g., how many voters preferred one candidate over the other).
The design of a mechanism that offers perfect ballot secrecy must be tailored to the specific voting rule that governs the elections. 
We demonstrated our solution on the following order-based rules: {\sc Copeland, Maximin, Kemeny-Young} and {\sc Modal-Ranking}.\footnote{The discussion of the latter two rules is deferred to the appendix.} 
Ours is the first study that offers a fully private solution for order-based voting rules that is lightweight and practical for elections in real-life large democracies. 

While several cryptographic voting protocols exist, they primarily focus on securing the casting and verification of ballots, rather than on computing complex voting outcomes under strict privacy constraints. To the best of our knowledge, there is currently no protocol that supports privacy-preserving computation of order-based voting rules with perfect ballot secrecy. Moreover, in related domains such as privacy-preserving collaborative filtering, prior art~\citep{ShmueliT20, TassaH22} shows that secure multiparty computation based on secret sharing is significantly more efficient than comparable approaches using homomorphic encryption. This performance gap reinforces the practicality of our approach for real-world elections. 

\subsection{An implementation of the proposed protocol and system in the Gentoo council elections}\label{gentoo}
The Gentoo Linux project (https://www.gentoo.org/) is a community-driven project with a
decentralized governance model. Its two main organizational units are
the Gentoo Foundation, which is a legal entity that holds trademarks,
domains, and other assets related to Gentoo, and the Gentoo Council---a seven-member elected body that is responsible for policy decisions as well as resolving disputes.

The Gentoo Council holds an annual election for selecting seven council members from a pool of candidates. During a two-week voting period, the voters submit ranked ballots to three election officials (talliers). These officials then perform a transparent tallying process, {\em without ballot secrecy}, and publish the results on the organization's mailing list.

Through private communication with Gentoo officials, we learned that the original election rule was
{\sc Schulze}  \citep{schulze2011new}. We demonstrated to the Gentoo officials our proposed voting model with a functional prototype implementation, which underwent a private evaluation. Following a successful vote by the election committee, they decided to adopt {\sc Copeland} rule instead of {\sc Schulze}'s, using our system as a pilot program for their next election cycle.

\subsection{Future research directions}\label{future}
The present study suggests several directions for future research:

(a) \textbf{Multi-winner elections.} Typical voting rules are usually oblivious to external social constraints, and they determine the identity of the winners solely based on private ballots.
This is the case with the voting rules that we considered herein. In contrast, {\em multi-winner} election rules are designed specifically for selecting candidates that would satisfy the voters the most~\citep{elkind2017multiwinner,faliszewski2017multiwinner}, in the sense that they also comply with additional social conditions (e.g., that the selected winners include a minimal number of representatives of a specific gender, race, region, etc.). This problem has unique features, and therefore requires its own secure protocols. 
Examples of voting rules that are designed for this purpose are {\sc Chamberlin-Courant}~\citep{chamberlin1983representative} and {\sc Monroe} ~\citep{monroe1995fully}.

(b) \textbf{A hierarchical tallier model.}
We assumed a ``flat'' tallier model, where all talliers are operating on all ballots. However, in large voting systems, a hierarchical tallier may be more suitable. For example, in the US, it may be more suitable to use a hierarchy by county (first level), state (second level), and national (third and highest level). A modification of our protocol for such settings is in order.

(c) \textbf{Malicious talliers.} 
Our protocol assumes that the talliers are semi-honest,
i.e., that they follow the prescribed protocol correctly. The semi-honesty of the talliers can be ensured in practice by securing the software and hardware of the talliers. However, it is possible to design an MPC protocol that would be immune even to malicious talliers that may deviate from the prescribed protocol. While such protocols are expected to have significantly higher runtimes, they could enhance even further the security of the system and the trust of voters in the preservation of their privacy.

\newpage
\appendix

\section{Secure computations over secret shared data}\label{scossd}

In Section \ref{principles} we described the sub-protocols that our protocol utilizes. Here we provide an inside look at each of those four MPC sub-protocols.

\subsection{Evaluating arithmetic functions}\label{mpc}
Let $f(x_1,\ldots,x_L)$ be an arithmetic function of $L$ variables. 
Assume that the talliers hold $D'$-out-of-$D$ shares of each of the inputs, $x_1,\ldots,x_L$. They wish to compute $D'$-out-of-$D$ shares of
$y:=f(x_1,\ldots,x_L)$ without learning any information on any of the inputs nor on the output $y$.
This challenge translates to the following two basic computational tasks: given $D'$-out-of-$D$ shares of two secret values $u,v \in \ZZ_p$, compute $D'$-out-of-$D$ 
shares
of $au+bv$ (where $a$ and $b$ are public field elements), and of $u \cdot v$, without revealing $u$, $v$ or the computed results ($au+bv$ and $u \cdot v$).

Let $u_d$ and $v_d$ denote the shares held by $T_d$, $d \in [D]$, corresponding to the two secrets $u$ and $v$, respectively.
Computing shares of the linear combination $au+bv$ is easy, since $\{ au_d+bv_d: d \in [D]\}$ is a valid $D'$-out-of-$D$ sharing of $au+bv$, as implied by the linearity of secret sharing. Consequently, each tallier can locally compute his share of the linear combination from his shares of the two inputs without any interaction with his peers.

Computing secret shares of the product $u \cdot v$ is more involved and requires the talliers to interact.
In our protocol, we use the multiplication protocol proposed by \citet{DN07}, enhanced by a work by \citet{ChidaGHIKLN18} that demonstrates some performance optimizations.
We consider that computation as a black-box since the details of that computation are not relevant for our needs. Interested readers may refer to \cite{DN07,ChidaGHIKLN18} for more details.

\subsection{Secure Comparison}\label{secure_comparison}
Assume that $T_1,\ldots,T_D$ hold $D'$-out-of-$D$ shares of two nonnegative integers $u$ and $v$, where both $u$ and $v$ are smaller than $p$, which is the size of the underlying field $\ZZ_p$. 
The goal is to compute
$D'$-out-of-$D$ secret shares of the bit $1_{u<v}$ without learning any other information on $u$ and $v$. 
Such a protocol is called secure comparison. In our protocol we used the secure comparison protocol proposed by Nishide and Ohta \citet{NO07}, with some performance enhancements that we introduced. 

\subsection{Secure testing of positivity}\label{positivity}
Let $u$ be an integer in the range $[-N,N]$.
Assume that $T_1,\ldots,T_D$ hold $D'$-out-of-$D$ shares of $u$, where the underlying field is $\ZZ_p$, and $p>2N$. Our goal is to design an MPC protocol that allows the talliers to obtain $D'$-out-of-$D$ shares of the bit $1_{u>0}$, without learning any additional information on $u$.

One way of solving such a problem would be to 
set $v=u+N$ and then test whether $N<v$ using the protocol of secure comparison from Section \ref{secure_comparison}. 
However, we propose a more efficient method for testing positivity. To this end, we state and prove the following lemma.

\begin{lemma}\label{lemmapositiv}
Under the above assumptions, $u>0$ iff the LSB of $(-2u \mod p)$ is 1.
\end{lemma}

\begin{proof}
Recall that $u \in [-N,N]$ and $N< \frac{p}{2}$.
Assume that $u>0$, namely, that $u \in (0,N]$. Hence, $-2u \in [-2N,0)$. Therefore, as $2N<p$, $(-2u \mod p) = -2u+p $. As that number is odd, its LSB is 1.
If, on the other hand, $u \leq 0$, then $u \in [-N,0]$. 
Hence, $-2u \in [0,2N] \subset [0,p-1]$. Therefore, $(-2u \mod p) = -2u $. As that number is even, its LSB is 0.
\end{proof}

Hence, the talliers can compute shares of $1_{u>0}$ by computing shares of the LSB of only one shared secret ($-2u$ in this case).
In order to compute shares of $1_{u<v}$ using the protocol of \cite{NO07} it is necessary to compute shares of three LSBs of shared secrets. Hence, computing shares of $1_{u>0}$ based on Lemma \ref{lemmapositiv} is more efficient, roughly by a factor of 3, than computing such shares based on the general secure comparison protocol of \cite{NO07}. In our protocol we compute shares of positivity bits based on this efficient computation.

\subsection{Secure testing of equality to zero}\label{is-zero}
Let $u$ be an integer in the range $[-N,N]$.
Assume that the talliers $T_1,\ldots,T_D$ hold $D'$-out-of-$D$ shares of $u$, where the underlying field is $\ZZ_p$, and $p>2N$. 
Our goal is to design an MPC protocol that enables the talliers to compute $D'$-out-of-$D$ shares of the bit $1_{u=0}$, without learning any additional information on $u$.

A na\"{i}ve approach would be to use the MPC positivity testing from Section \ref{positivity}, once for $u$ and once for $-u$. Clearly, $u=0$ iff both of those tests fail. However, we can solve that problem in a more efficient manner, as we proceed to describe.

Fermat's little theorem states that if $u \in \ZZ_p \setminus \{0\}$ then $u^{p-1}=1$ mod $p$. Hence, $ 1_{u\neq 0} = \left( u^{p-1} \mod p \right)$. Therefore, shares of the bit $ 1_{u\neq 0}$ can be obtained by simply computing $u^{p-1}$ mod $p$. The latter computation can be carried out by the square-and-multiply algorithm with up to $2 \lceil \log p \rceil$ consecutive multiplications.
Finally, as $1_{u = 0} = 1 -1_{ u \neq 0}$, then shares of $1_{ u \neq 0}$ can be readily translated into shares of $1_{u=0}$. The cost of the above described computation is significantly smaller than the cost of the alternative approach that performs positivity testing of both $u$ and $-u$.

\section{Other order-based rules}\label{other}
Here we consider two other order-based rules $-$ {\sc Kemeny-Young} ~\citep{kemeny1959mathematics,young1995optimal}
and {\sc Modal Ranking}~\citep{caragiannis2014modal}, and describe secure protocols for implementing them.

\subsection{Kemeny-Young}
\subsubsection{Description}

The ballot matrices here, $P_n$, $n \in [N]$, are as in {\sc Maximin} and, as before, $P=\sum_{n \in [N]}P_n$ is the aggregated ballot matrix.
For every $1 \leq m \neq \ell \leq N$, $P(m,\ell)$ equals the number of voters who ranked $C_m$ strictly higher than $C_\ell$.
The matrix $P$ induces a score for each of the possible $M!$ rankings over
$\bC = \{C_1,\ldots,C_M\}$. Let $\rho = (\sigma_1,\ldots,\sigma_M)$, where $(\sigma_1,\ldots,\sigma_M)$ is a permutation of $\{1,\ldots,M\}$, be such a ranking. Here, for each $m \in [M]$, $\sigma_m$ is the rank of $C_m$ in the ranking. For example, if $M=4$ then 
$\rho = (3,1,4,2)$ is the ranking in which $C_2$ is the top candidate and $C_3$ is the least favored candidate. The score of a ranking $\rho$ is defined as
\be w(\rho) = \sum_{\ell,m \in [M]: \sigma_m<\sigma_\ell} P(m,\ell) \,;
\label{rhoscore}\ee
namely, one goes over all pairs of candidates $C_m$ and
$C_\ell$ such that $\rho$ ranks $C_m$ higher than $C_\ell$ (in the sense that $\sigma_m<\sigma_\ell$) and adds up the number of voters who agreed with this pairwise comparison. The ranking $\rho$ with the highest score is selected, and the $K$ leading candidates in $\rho$ are the winners of the election.

\subsubsection{Secure implementation}
As in Protocol \ref{algvote} for {\sc Copeland} and {\sc Maximin}, each voter $V_n$ secret shares the entries of his ballot matrix $P_n$ among the $D$ talliers using a $D'$-out-of-$D$ scheme. Since the ballot matrices here are as in {\sc Maximin}, where ties are allowed, their validation is carried out as described in Section \ref{maximinties}.

After validating the cast ballots, the talliers add up their shares of $P_n$, for all validated ballot matrices, and then get $D'$-out-of-$D$ shares of each of the non-diagonal entries in $P$.

To compute secret shares of the score of each possible ranking $\rho$, the talliers need only to perform summation according to Eq. (\ref{rhoscore}). Note that that computation does not require the talliers to interact. Finally, it is needed to find the ranking with the highest score. That computation can be done by performing secure comparisons, as described in Section \ref{secure_comparison}.

\subsection{Modal ranking}
In {\sc Modal Ranking}, every voter submits a ranking of all candidates, where the ranking has no ties. Namely, if $R_{\bC}$ is the set of all $M!$ rankings of the $M$ candidates in $\bC$, the ballot of each voter is a selection of one ranking from $R_{\bC}$, as determined by his preferences. The rule then outputs the ranking that was selected by the largest number of voters (i.e., the mode of the distribution of ballots over $R_{\bC}$). In case there are several rankings that were selected by the greatest number of voters, the rule outputs all of them, and then the winners are usually the candidates whose average position in those rankings is the highest.

The {\sc Modal Ranking} rule is an order-based voting rule over $\bC$, but it is equivalent to the {\sc Plurality} score-based voting rule (see~\citet{brandt2016handbook}) over the set of candidate rankings $R_{\bC}$. Hence, it can be securely implemented 
by the protocol that was presented in \citet{dery2021fear}.

\newpage
\bibliographystyle{plain}

\begin{thebibliography}{46}


\ifx \showCODEN    \undefined \def \showCODEN     #1{\unskip}     \fi
\ifx \showDOI      \undefined \def \showDOI       #1{#1}\fi
\ifx \showISBNx    \undefined \def \showISBNx     #1{\unskip}     \fi
\ifx \showISBNxiii \undefined \def \showISBNxiii  #1{\unskip}     \fi
\ifx \showISSN     \undefined \def \showISSN      #1{\unskip}     \fi
\ifx \showLCCN     \undefined \def \showLCCN      #1{\unskip}     \fi
\ifx \shownote     \undefined \def \shownote      #1{#1}          \fi
\ifx \showarticletitle \undefined \def \showarticletitle #1{#1}   \fi
\ifx \showURL      \undefined \def \showURL       {\relax}        \fi
\providecommand\bibfield[2]{#2}
\providecommand\bibinfo[2]{#2}
\providecommand\natexlab[1]{#1}
\providecommand\showeprint[2][]{arXiv:#2}


\bibitem[Adida(2008)]%
        {adida2008helios}
\bibfield{author}{\bibinfo{person}{Ben Adida}.}
  \bibinfo{year}{2008}\natexlab{}.
\newblock \showarticletitle{Helios: Web-based Open-Audit Voting.}. In
  \bibinfo{booktitle}{\emph{USENIX security symposium}},
  Vol.~\bibinfo{volume}{17}. \bibinfo{pages}{335--348}.
\newblock


\bibitem[Balinski and Laraki(2007)]%
        {balinski2007theory}
\bibfield{author}{\bibinfo{person}{Michel Balinski} {and} \bibinfo{person}{Rida
  Laraki}.} \bibinfo{year}{2007}\natexlab{}.
\newblock \showarticletitle{A theory of measuring, electing, and ranking}.
\newblock \bibinfo{journal}{\emph{Proceedings of the National Academy of
  Sciences}} \bibinfo{volume}{104}, \bibinfo{number}{21}
  (\bibinfo{year}{2007}), \bibinfo{pages}{8720--8725}.
\newblock


\bibitem[Benaloh(1986)]%
        {Benaloh86a}
\bibfield{author}{\bibinfo{person}{Josh Benaloh}.}
  \bibinfo{year}{1986}\natexlab{}.
\newblock \showarticletitle{Secret sharing homomorphisms: Keeping shares of a
  secret secret}. In \bibinfo{booktitle}{\emph{CRYPTO}}.
  \bibinfo{pages}{251--260}.
\newblock

\bibitem{BogdanovLW08}
Dan Bogdanov, Sven Laur, and Jan Willemson.
2008. Sharemind: A framework for fast privacy-preserving computations. In \emph{{ESORICS}}. 192--206.
\newblock


\bibitem[Blakley(1979)]%
        {Blakley}
\bibfield{author}{\bibinfo{person}{G.R. Blakley}.}
  \bibinfo{year}{1979}\natexlab{}.
\newblock \showarticletitle{Safeguarding Cryptographic Keys}. In
  \bibinfo{booktitle}{\emph{International Workshop on Managing Requirements
  Knowledge}}. \bibinfo{pages}{313--317}.
\newblock


\bibitem[Boneh and Golle(2002)]%
        {boneh2002almost}
\bibfield{author}{\bibinfo{person}{Dan Boneh} {and} \bibinfo{person}{Philippe
  Golle}.} \bibinfo{year}{2002}\natexlab{}.
\newblock \showarticletitle{Almost entirely correct mixing with applications to
  voting}. In \bibinfo{booktitle}{\emph{CCS}}. \bibinfo{pages}{68--77}.
\newblock


\bibitem[Brandt et~al\mbox{.}(2016)]%
        {brandt2016handbook}
\bibfield{author}{\bibinfo{person}{Felix Brandt}, \bibinfo{person}{Vincent
  Conitzer}, \bibinfo{person}{Ulle Endriss}, \bibinfo{person}{J{\'e}r{\^o}me
  Lang}, {and} \bibinfo{person}{Ariel~D Procaccia}.}
  \bibinfo{year}{2016}\natexlab{}.
\newblock \bibinfo{booktitle}{\emph{Handbook of computational social choice}}.
\newblock \bibinfo{publisher}{Cambridge University Press}.
\newblock


\bibitem[Brandt and Sandholm(2005)]%
        {brandt2005decentralized}
\bibfield{author}{\bibinfo{person}{Felix Brandt} {and} \bibinfo{person}{Tuomas
  Sandholm}.} \bibinfo{year}{2005}\natexlab{}.
\newblock \showarticletitle{Decentralized voting with unconditional privacy}.
  In \bibinfo{booktitle}{\emph{AAMAS}}. \bibinfo{pages}{357--364}.
\newblock


\bibitem[Canard et~al\mbox{.}(2018)]%
        {canard2018practical}
\bibfield{author}{\bibinfo{person}{S{\'e}bastien Canard},
  \bibinfo{person}{David Pointcheval}, \bibinfo{person}{Quentin Santos}, {and}
  \bibinfo{person}{Jacques Traor{\'e}}.} \bibinfo{year}{2018}\natexlab{}.
\newblock \showarticletitle{Practical strategy-resistant privacy-preserving
  elections}. In \bibinfo{booktitle}{\emph{European Symposium on Research in
  Computer Security}}. \bibinfo{pages}{331--349}.
\newblock


\bibitem[Caragiannis et~al\mbox{.}(2014)]%
        {caragiannis2014modal}
\bibfield{author}{\bibinfo{person}{Ioannis Caragiannis},
  \bibinfo{person}{Ariel~D Procaccia}, {and} \bibinfo{person}{Nisarg Shah}.}
  \bibinfo{year}{2014}\natexlab{}.
\newblock \showarticletitle{Modal ranking: A uniquely robust voting rule}. In
  \bibinfo{booktitle}{\emph{AAAI}}.
\bibinfo{pages}{616-622}.
\newblock


\bibitem[Chamberlin and Courant(1983)]%
        {chamberlin1983representative}
\bibfield{author}{\bibinfo{person}{John~R Chamberlin} {and}
  \bibinfo{person}{Paul~N Courant}.} \bibinfo{year}{1983}\natexlab{}.
\newblock \showarticletitle{Representative deliberations and representative
  decisions: Proportional representation and the Borda rule}.
\newblock \bibinfo{journal}{\emph{The American Political Science Review}}
  (\bibinfo{year}{1983}), \bibinfo{pages}{718--733}.
\newblock


\bibitem[Chaum(1988)]%
        {chaum1988}
\bibfield{author}{\bibinfo{person}{David Chaum}.}
  \bibinfo{year}{1988}\natexlab{}.
\newblock \showarticletitle{Elections with Unconditionally-Secret Ballots and
  Disruption Equivalent to Breaking {RSA}}. In
  \bibinfo{booktitle}{\emph{EUROCRYPT}}. \bibinfo{pages}{177--182}.
\newblock


\bibitem[Chaum(1981)]%
        {chaum1981untraceable}
\bibfield{author}{\bibinfo{person}{David~L Chaum}.}
  \bibinfo{year}{1981}\natexlab{}.
\newblock \showarticletitle{Untraceable electronic mail, return addresses, and
  digital pseudonyms}.
\newblock \bibinfo{journal}{\emph{Commun. {ACM}}} \bibinfo{volume}{24},
  \bibinfo{number}{2} (\bibinfo{year}{1981}), \bibinfo{pages}{84--90}.
\newblock

\bibitem[Chida et~al\mbox{.}(2018)]%
        {ChidaGHIKLN18}
\bibfield{author}{\bibinfo{person}{Koji Chida}, \bibinfo{person}{Daniel
  Genkin}, \bibinfo{person}{Koki Hamada}, \bibinfo{person}{Dai Ikarashi},
  \bibinfo{person}{Ryo Kikuchi}, \bibinfo{person}{Yehuda Lindell}, {and}
  \bibinfo{person}{Ariel Nof}.} \bibinfo{year}{2018}\natexlab{}.
\newblock \showarticletitle{Fast Large-Scale Honest-Majority {MPC} for
  Malicious Adversaries}. In \bibinfo{booktitle}{\emph{{CRYPTO}}}.
  \bibinfo{pages}{34--64}.
\newblock


\bibitem[Copeland(1951)]%
        {copeland1951reasonable}
\bibfield{author}{\bibinfo{person}{Arthur~H Copeland}.}
  \bibinfo{year}{1951}\natexlab{}.
\newblock \showarticletitle{A reasonable social welfare function}. In
  \bibinfo{booktitle}{\emph{Mimeographed notes from a Seminar on Applications
  of Mathematics to the Social Sciences, University of Michigan}}.
\newblock


\bibitem[Cortier et~al\mbox{.}(2022)]%
        {cortier2022toolbox}
\bibfield{author}{\bibinfo{person}{V{\'e}ronique Cortier},
  \bibinfo{person}{Pierrick Gaudry}, {and} \bibinfo{person}{Quentin Yang}.}
  \bibinfo{year}{2022}\natexlab{}.
\newblock \showarticletitle{A toolbox for verifiable tally-hiding e-voting
  systems}. In \bibinfo{booktitle}{\emph{European Symposium on Research in
  Computer Security}}. \bibinfo{pages}{631--652}.
\newblock


\bibitem[Cramer et~al\mbox{.}(1997)]%
        {cramer1997secure}
\bibfield{author}{\bibinfo{person}{Ronald Cramer}, \bibinfo{person}{Rosario
  Gennaro}, {and} \bibinfo{person}{Berry Schoenmakers}.}
  \bibinfo{year}{1997}\natexlab{}.
\newblock \showarticletitle{A secure and optimally efficient multi-authority
  election scheme}. In \bibinfo{booktitle}{\emph{EUROCRYPT}}.
  \bibinfo{pages}{103--118}.
\newblock


\bibitem[Damg{\aa}rd et~al\mbox{.}(2010)]%
        {damgaard2010generalization}
\bibfield{author}{\bibinfo{person}{Ivan Damg{\aa}rd}, \bibinfo{person}{Mads
  Jurik}, {and} \bibinfo{person}{Jesper~Buus Nielsen}.}
  \bibinfo{year}{2010}\natexlab{}.
\newblock \showarticletitle{A generalization of Pailliers public-key system
  with applications to electronic voting}.
\newblock \bibinfo{journal}{\emph{International Journal of Information
  Security}} \bibinfo{volume}{9} 
 (\bibinfo{year}{2010}),
  \bibinfo{pages}{371--385}.
\newblock


\bibitem[Damg{\aa}rd and Nielsen(2007)]%
        {DN07}
\bibfield{author}{\bibinfo{person}{Ivan Damg{\aa}rd} {and}
  \bibinfo{person}{Jesper~Buus Nielsen}.} \bibinfo{year}{2007}\natexlab{}.
\newblock \showarticletitle{Scalable and Unconditionally Secure Multiparty
  Computation}. In \bibinfo{booktitle}{\emph{{CRYPTO}}}.
  \bibinfo{pages}{572--590}.
\newblock


\bibitem[David(1963)]%
        {david1963method}
\bibfield{author}{\bibinfo{person}{Herbert~Aron David}.}
  \bibinfo{year}{1963}\natexlab{}.
\newblock \bibinfo{booktitle}{\emph{The method of paired comparisons}}.
  Vol.~\bibinfo{volume}{12}.
\newblock \bibinfo{publisher}{London}.
\newblock


\bibitem[Dery et~al\mbox{.}(2021a)]%
        {dery2021fear}
\bibfield{author}{\bibinfo{person}{Lihi Dery}, \bibinfo{person}{Tamir Tassa},
  {and} \bibinfo{person}{Avishay Yanai}.} \bibinfo{year}{2021}\natexlab{a}.
\newblock \showarticletitle{Fear not, vote truthfully: Secure Multiparty
  Computation of score based rules}.
\newblock \bibinfo{journal}{\emph{Expert Systems with Applications}}
  \bibinfo{volume}{168} (\bibinfo{year}{2021}), \bibinfo{pages}{114434}.
\newblock


\bibitem[Dery et~al\mbox{.}(2021b)]%
        {dery2021secure}
\bibfield{author}{\bibinfo{person}{Lihi Dery}, \bibinfo{person}{Tamir Tassa},
  \bibinfo{person}{Avishay Yanai}, {and} \bibinfo{person}{Arthur Zamarin}.}
  \bibinfo{year}{2021}\natexlab{b}.
\newblock \showarticletitle{A Secure Voting System for Score Based Elections}.
  In \bibinfo{booktitle}{\emph{CCS}}. \bibinfo{pages}{2399--2401}.
\newblock


\bibitem[Elkind et~al\mbox{.}(2017)]%
        {elkind2017multiwinner}
\bibfield{author}{\bibinfo{person}{Edith Elkind}, \bibinfo{person}{Piotr
  Faliszewski}, \bibinfo{person}{Jean-Fran{\c{c}}ois Laslier},
  \bibinfo{person}{Piotr Skowron}, \bibinfo{person}{Arkadii Slinko}, {and}
  \bibinfo{person}{Nimrod Talmon}.} \bibinfo{year}{2017}\natexlab{}.
\newblock \showarticletitle{What Do Multiwinner Voting Rules Do? An Experiment
  Over the Two-Dimensional Euclidean Domain}. In
  \bibinfo{booktitle}{\emph{AAAI}}.
 \bibinfo{pages}{494-501}.
\newblock


\bibitem[Faliszewski et~al\mbox{.}(2009)]%
        {faliszewski2009llull}
\bibfield{author}{\bibinfo{person}{Piotr Faliszewski}, \bibinfo{person}{Edith
  Hemaspaandra}, \bibinfo{person}{Lane~A Hemaspaandra}, {and}
  \bibinfo{person}{J{\"o}rg Rothe}.} \bibinfo{year}{2009}\natexlab{}.
\newblock \showarticletitle{Llull and Copeland voting computationally resist
  bribery and constructive control}.
\newblock \bibinfo{journal}{\emph{Journal of Artificial Intelligence Research}}
  \bibinfo{volume}{35} (\bibinfo{year}{2009}),
  \bibinfo{pages}{275--341}.
\newblock


\bibitem[Faliszewski et~al\mbox{.}(2017)]%
        {faliszewski2017multiwinner}
\bibfield{author}{\bibinfo{person}{Piotr Faliszewski}, \bibinfo{person}{Piotr
  Skowron}, \bibinfo{person}{Arkadii Slinko}, {and} \bibinfo{person}{Nimrod
  Talmon}.} \bibinfo{year}{2017}\natexlab{}.
\newblock \showarticletitle{Multiwinner voting: A new challenge for social
  choice theory}.
\newblock \bibinfo{journal}{\emph{Trends in computational social choice}}
  \bibinfo{volume}{74} (\bibinfo{year}{2017}), \bibinfo{pages}{27--47}.
\newblock


\bibitem[Fan et~al\mbox{.}(2020)]%
        {fan2020hse}
\bibfield{author}{\bibinfo{person}{Xingyue Fan}, \bibinfo{person}{Ting Wu},
  \bibinfo{person}{Qiuhua Zheng}, \bibinfo{person}{Yuanfang Chen},
  \bibinfo{person}{Muhammad Alam}, {and} \bibinfo{person}{Xiaodong Xiao}.}
  \bibinfo{year}{2020}\natexlab{}.
\newblock \showarticletitle{HSE-Voting: A secure high-efficiency electronic
  voting scheme based on homomorphic signcryption}.
\newblock \bibinfo{journal}{\emph{Future Generation Computer Systems}}
  \bibinfo{volume}{111} (\bibinfo{year}{2020}), \bibinfo{pages}{754--762}.
\newblock


\bibitem[Haines et~al\mbox{.}(2020)]%
{haines2020verifiable}
Thomas Haines, Dirk Pattinson, and Mukesh Tiwari.
\newblock Verifiable homomorphic tallying for the schulze vote counting scheme.
\newblock In {\em Verified Software. Theories, Tools, and Experiments: 11th International Conference, VSTTE 2019, New York City, NY, USA, July 13--14, 2019, Revised Selected Papers 11}, pages 36--53. Springer, 2020.



\bibitem[Hertel et~al\mbox{.}(2021)]%
        {hertel2021extending}
\bibfield{author}{\bibinfo{person}{Fabian Hertel}, \bibinfo{person}{Nicolas
  Huber}, \bibinfo{person}{Jonas Kittelberger}, \bibinfo{person}{Ralf
  K{\"u}sters}, \bibinfo{person}{Julian Liedtke}, {and} \bibinfo{person}{Daniel
  Rausch}.} \bibinfo{year}{2021}\natexlab{}.
\newblock \showarticletitle{Extending the Tally-Hiding Ordinos System:
  Implementations for Borda, Hare-Niemeyer, Condorcet, and Instant-Runoff
  Voting}. In \bibinfo{booktitle}{\emph{E-Vote-ID 2021}}. University of Tartu
  Press, \bibinfo{pages}{269--284}.
\newblock


\bibitem[Hevia and Kiwi(2004)]%
        {hevia2004electronic}
\bibfield{author}{\bibinfo{person}{Alejandro Hevia} {and}
  \bibinfo{person}{Marcos Kiwi}.} \bibinfo{year}{2004}\natexlab{}.
\newblock \showarticletitle{Electronic jury voting protocols}.
\newblock \bibinfo{journal}{\emph{Theoretical Computer Science}}
  \bibinfo{volume}{321} (\bibinfo{year}{2004}),
  \bibinfo{pages}{73--94}.
\newblock


\bibitem[Jakobsson et~al\mbox{.}(2002)]%
        {jakobsson2002making}
\bibfield{author}{\bibinfo{person}{Markus Jakobsson}, \bibinfo{person}{Ari
  Juels}, {and} \bibinfo{person}{Ronald~L Rivest}.}
  \bibinfo{year}{2002}\natexlab{}.
\newblock \showarticletitle{Making mix nets robust for electronic voting by
  randomized partial checking.} In \bibinfo{booktitle}{\emph{USENIX}}.  \bibinfo{pages}{339--353}.
\newblock


\bibitem[Kemeny(1959)]%
        {kemeny1959mathematics}
\bibfield{author}{\bibinfo{person}{John~G Kemeny}.}
  \bibinfo{year}{1959}\natexlab{}.
\newblock \showarticletitle{Mathematics without numbers}.
\newblock \bibinfo{journal}{\emph{Daedalus}} \bibinfo{volume}{88},
  \bibinfo{number}{4} (\bibinfo{year}{1959}), \bibinfo{pages}{577--591}.
\newblock


\bibitem[K{\"u}sters et~al\mbox{.}(2020)]%
        {kusters2020ordinos}
\bibfield{author}{\bibinfo{person}{Ralf K{\"u}sters}, \bibinfo{person}{Julian
  Liedtke}, \bibinfo{person}{Johannes M{\"u}ller}, \bibinfo{person}{Daniel
  Rausch}, {and} \bibinfo{person}{Andreas Vogt}.}
  \bibinfo{year}{2020}\natexlab{}.
\newblock \showarticletitle{Ordinos: A verifiable tally-hiding e-voting
  system}. In \bibinfo{booktitle}{\emph{IEEE European Symposium on
  Security and Privacy (EuroS\&P)}}.  \bibinfo{pages}{216--235}.
\newblock

\bibitem{lau2001advantages}
Richard~R Lau and David~P Redlawsk.
\newblock Advantages and disadvantages of cognitive heuristics in political decision making.
\newblock {\em American journal of political science}, pages 951--971, 2001.

\bibitem[Lee et~al\mbox{.}(2003)]%
        {lee2003providing}
\bibfield{author}{\bibinfo{person}{Byoungcheon Lee}, \bibinfo{person}{Colin
  Boyd}, \bibinfo{person}{Ed Dawson}, \bibinfo{person}{Kwangjo Kim},
  \bibinfo{person}{Jeongmo Yang}, {and} \bibinfo{person}{Seungjae Yoo}.}
  \bibinfo{year}{2003}\natexlab{}.
\newblock \showarticletitle{Providing receipt-freeness in mixnet-based voting
  protocols}. In \bibinfo{booktitle}{\emph{International conference on
  information security and cryptology}}. \bibinfo{pages}{245--258}.
\newblock

\bibitem{miller1956magical}
George~A Miller.
\newblock The magical number seven, plus or minus two: Some limits on our capacity for processing information.
\newblock {\em Psychological review}, 63(2):81, 1956.


\bibitem[Monroe(1995)]%
        {monroe1995fully}
\bibfield{author}{\bibinfo{person}{Burt~L Monroe}.}
  \bibinfo{year}{1995}\natexlab{}.
\newblock \showarticletitle{Fully proportional representation}.
\newblock \bibinfo{journal}{\emph{American Political Science Review}}
  (\bibinfo{year}{1995}), \bibinfo{pages}{925--940}.
\newblock


\bibitem[Nair et~al\mbox{.}(2015)]%
        {NairBK15}
\bibfield{author}{\bibinfo{person}{Divya~G. Nair}, \bibinfo{person}{V.~P.
  Binu}, {and} \bibinfo{person}{G.~Santhosh Kumar}.}
  \bibinfo{year}{2015}\natexlab{}.
\newblock \showarticletitle{An Improved E-voting scheme using Secret Sharing
  based Secure Multi-party Computation}.
\newblock \bibinfo{journal}{\emph{CoRR}}  \bibinfo{volume}{abs/1502.07469}
  (\bibinfo{year}{2015}).
\newblock


\bibitem[Neff(2001)]%
        {neff2001verifiable}
\bibfield{author}{\bibinfo{person}{C~Andrew Neff}.}
  \bibinfo{year}{2001}\natexlab{}.
\newblock \showarticletitle{A verifiable secret shuffle and its application to
  e-voting}. In \bibinfo{booktitle}{\emph{CCS}}. \bibinfo{pages}{116--125}.
\newblock


\bibitem[Nishide and Ohta(2007)]%
        {NO07}
\bibfield{author}{\bibinfo{person}{Takashi Nishide} {and}
  \bibinfo{person}{Kazuo Ohta}.} \bibinfo{year}{2007}\natexlab{}.
\newblock \showarticletitle{Multiparty Computation for Interval, Equality, and
  Comparison Without Bit-Decomposition Protocol}. In
  \bibinfo{booktitle}{\emph{{PKC}}}. \bibinfo{pages}{343--360}.
\newblock


\bibitem[Priya et~al\mbox{.}(2018)]%
        {priya2018blockchain}
\bibfield{author}{\bibinfo{person}{J~Chandra Priya}, \bibinfo{person}{Ponsy
  RK~Sathia Bhama}, \bibinfo{person}{S Swarnalaxmi},
  \bibinfo{person}{A~Aisathul Safa}, {and} \bibinfo{person}{I Elakkiya}.}
  \bibinfo{year}{2018}\natexlab{}.
\newblock \showarticletitle{Blockchain centered homomorphic encryption: A
  secure solution for E-balloting}. In \bibinfo{booktitle}{\emph{International
  conference on Computer Networks, Big data and IoT}}. \bibinfo{pages}{811--819}.
\newblock


\bibitem[Rezaeibagha et~al\mbox{.}(2019)]%
        {rezaeibagha2019provably}
\bibfield{author}{\bibinfo{person}{Fatemeh Rezaeibagha}, \bibinfo{person}{Yi
  Mu}, \bibinfo{person}{Shiwei Zhang}, {and} \bibinfo{person}{Xiaofen Wang}.}
  \bibinfo{year}{2019}\natexlab{}.
\newblock \showarticletitle{Provably secure (broadcast) homomorphic
  signcryption}.
\newblock \bibinfo{journal}{\emph{International Journal of Foundations of
  Computer Science}} \bibinfo{volume}{30}, 
  (\bibinfo{year}{2019}), \bibinfo{pages}{511--529}.
\newblock


\bibitem[Sako and Kilian(1995)]%
        {sako1995receipt}
\bibfield{author}{\bibinfo{person}{Kazue Sako} {and} \bibinfo{person}{Joe
  Kilian}.} \bibinfo{year}{1995}\natexlab{}.
\newblock \showarticletitle{Receipt-free mix-type voting scheme}. In
  \bibinfo{booktitle}{\emph{International Conference on the Theory and
  Applications of Cryptographic Techniques}}.
  \bibinfo{pages}{393--403}.
\newblock


\bibitem[Schulze(2011)]%
        {schulze2011new}
\bibfield{author}{\bibinfo{person}{Markus Schulze}.}
  \bibinfo{year}{2011}\natexlab{}.
\newblock \showarticletitle{A new monotonic, clone-independent, reversal
  symmetric, and condorcet-consistent single-winner election method}.
\newblock \bibinfo{journal}{\emph{Social choice and Welfare}}
  \bibinfo{volume}{36},  (\bibinfo{year}{2011}),
  \bibinfo{pages}{267--303}.
\newblock


\bibitem[Shamir(1979)]%
        {Shamir79}
\bibfield{author}{\bibinfo{person}{Adi Shamir}.}
  \bibinfo{year}{1979}\natexlab{}.
\newblock \showarticletitle{How to Share a Secret}.
\newblock \bibinfo{journal}{\emph{Commun. {ACM}}}  \bibinfo{volume}{22}
  (\bibinfo{year}{1979}), \bibinfo{pages}{612--613}.
\newblock

\bibitem{ShmueliT20}
Erez Shmueli and Tamir Tassa.
\newblock Mediated secure multi-party protocols for collaborative filtering.
\newblock {\em {ACM} Trans. Intell. Syst. Technol.}, 11:15:1--15:25, 2020.

\bibitem{TassaH22}
Tamir Tassa and Alon~Ben Horin.
\newblock Privacy-preserving collaborative filtering by distributed mediation.
\newblock {\em {ACM} Trans. Intell. Syst. Technol.}, 13:102:1--102:26, 2022.

\bibitem[Simpson(1969)]%
        {simpson1969defining}
\bibfield{author}{\bibinfo{person}{Paul~B Simpson}.}
  \bibinfo{year}{1969}\natexlab{}.
\newblock \showarticletitle{On defining areas of voter choice: Professor
  Tullock on stable voting}.
\newblock \bibinfo{journal}{\emph{The Quarterly Journal of Economics}}
  (\bibinfo{year}{1969}), \bibinfo{pages}{478--490}.
\newblock


\bibitem[Szepieniec and Preneel(2015)]%
        {szepieniec2015new}
\bibfield{author}{\bibinfo{person}{Alan Szepieniec} {and} \bibinfo{person}{Bart
  Preneel}.} \bibinfo{year}{2015}\natexlab{}.
\newblock \showarticletitle{New techniques for electronic voting}. In
  \bibinfo{booktitle}{\emph{{IACR} Cryptol. ePrint Arch.}}. \bibinfo{pages}{809}.
\newblock


\bibitem[Yang et~al\mbox{.}(2018)]%
        {yang2018secure}
\bibfield{author}{\bibinfo{person}{Xuechao Yang}, \bibinfo{person}{Xun Yi},
  \bibinfo{person}{Surya Nepal}, \bibinfo{person}{Andrei Kelarev}, {and}
  \bibinfo{person}{Fengling Han}.} \bibinfo{year}{2018}\natexlab{}.
\newblock \showarticletitle{A secure verifiable ranked choice online voting
  system based on homomorphic encryption}.
\newblock \bibinfo{journal}{\emph{IEEE Access}}  \bibinfo{volume}{6}
  (\bibinfo{year}{2018}), \bibinfo{pages}{20506--20519}.
\newblock


\bibitem[Yao(1982)]%
        {yao}
\bibfield{author}{\bibinfo{person}{A.C. Yao}.} \bibinfo{year}{1982}\natexlab{}.
\newblock \showarticletitle{{Protocols for secure computation}}. In
  \bibinfo{booktitle}{\emph{FOCS}}. \bibinfo{pages}{160--164}.
\newblock


\bibitem[Young(1995)]%
        {young1995optimal}
\bibfield{author}{\bibinfo{person}{Peyton Young}.}
  \bibinfo{year}{1995}\natexlab{}.
\newblock \showarticletitle{Optimal voting rules}.
\newblock \bibinfo{journal}{\emph{Journal of Economic Perspectives}}
  \bibinfo{volume}{9},  (\bibinfo{year}{1995}),
  \bibinfo{pages}{51--64}.
\newblock


\end{thebibliography}

\end{document}